\documentclass[11pt,a4paper]{amsart}
\usepackage[utf8]{inputenc}
\usepackage{amsmath}
\usepackage{amsfonts}
\usepackage{amssymb}
\usepackage{graphicx}
\usepackage{booktabs}
\usepackage[left=2.2cm,right=2.2cm,top=2cm,bottom=2cm]{geometry}
\usepackage{xcolor}
\usepackage[ruled,vlined]{algorithm2e}
\usepackage{changes}
\definechangesauthor[name=Julian, color=blue]{j}
\definechangesauthor[name=Ariel, color=orange]{a}
\numberwithin{equation}{section}  
\newtheorem{defn}{Definition}[section]
\newtheorem{exa}[defn]{Example}
\newtheorem{rem}[defn]{Remark}

\newtheorem{thm}[defn]{Theorem}

\newtheorem{asu}[defn]{Assumption}

\newcommand{\D}{\,{\mathop{}\!\mathrm{d}}} 
\newcommand{\R}{\mathbb{R}}
\newcommand{\Q}{\mathbb{Q}}

\newcommand{\K}{\mathbb{K}}
\newcommand{\N}{\mathbb{N}}

\newcommand{\E}{\mathbb{E}}

\newcommand{\one}{ 1 \hspace{-3pt} \mathrm{l}} %

\usepackage{hyperref}
\usepackage{xcolor}
\hypersetup{
    colorlinks,
    linkcolor={red!50!black},
    citecolor={blue!50!black},
    urlcolor={blue!80!black}
}

\title{Model-free price bounds under dynamic option trading}
\author{Ariel Neufeld \and Julian Sester}
\address{NTU Singapore, Division of Mathematical Sciences; 21 Nanyang Link, Singapore 637371}

\begin{document}

\begin{abstract}
In this paper we extend discrete time semi-static trading strategies by also allowing for dynamic trading in a finite amount of options, and we study the consequences for the model-independent super-replication prices of exotic derivatives. These include duality results as well as a precise characterization of pricing rules for the dynamically tradable options triggering an improvement of the price bounds for exotic derivatives in comparison with the conventional price bounds obtained through the martingale optimal transport approach.\\

\end{abstract}

\maketitle

\section{Introduction}
In practice it is a well-known and often faced problem that given a specific market model, super-hedging strategies for financial derivatives are very expensive to implement (compare \cite{biagini2004super}, \cite[Example 7.21]{follmer2016stochastic}, \cite{frey1999bounds} and \cite{neufeld2018buy} for examples in absence of transaction costs. Further we refer to \cite{cvitanic1999closed} and \cite{levental1997possibility} for examples in the presence of transaction costs). If an investor is interested in considering hedging strategies that super-replicate the payoff of a derivative under parameter-uncertainty of a specific model class or even completely model-independent, then prices for super-hedges become even higher than model-specific hedges, compare for this the model-independent approaches from \cite{beiglbock2013model}, \cite{dolinsky2014martingale} as well as the robust approach from \cite{bouchard2015arbitrage}. It has therefore become an agile recent research topic to reduce model-independent super-hedging prices through the introduction of different properties of models inferred from financial markets or via the inclusion of additional information (compare e.g. \cite{ansari2020improved}, \cite{aquino2019bounds}, \cite{eckstein2020martingale}, \cite{lutkebohmert2019tightening}, and \cite{sester2020robust}).

We contribute to the literature on model-free pricing by studying a novel approach to reduce prices of model-independent super-replication strategies. The approach relies on extending the usually considered class of semi-static trading strategies by additionally allowing dynamic trading in a finite amount of liquid European options. As we will show, this larger class of trading strategies may then allow to realize more flexible payoffs and thus to reduce super-hedging prices in many situations.

In an $n$-period financial market model, the assumption of being able to trade in European call options with expiration date $t_j$ (and other liquid kind of options) not only at initial time $t_0$, but also at intermediate times $(t_i)_{i=1,\dots,j-1}$ can be motivated by an observation that can be made on many real financial markets. Usually, expiration dates of European call options are bound to specific dates\footnote{For many call options this is the last trading day before the 20th of a month.}, i.e., if options with maturity $t_j$ are assumed to be liquidly traded at time $t_0$, then there are also quotes available for the same maturity (and a shorter time-to-maturity) at time $t_i>t_0$ with $t_i <t_j$, and the options can also be assumed to be tradeable at the later date. Thus, from a practical point of view, it seems natural to consider strategies incorporating dynamic trading in the underlying security as well as in European options.

It is one of the fundamental ideas in mathematical finance that minimizing prices over specific classes of super-replication strategies yields in many situations the same value as the maximal (risk-neutral) expectation of the payoff that is super-replicated, where the expectation is taken w.r.t.\ (martingale) measures from a specific model-class which is strongly related to the class of admissible super-hedges. The precise mathematical formulation of this result is known as super-hedging duality and it can be derived in different settings, compare \cite{acciaio2016model}, \cite{aksamit2019robust}, \cite{bayraktar2017arbitrage},  \cite{beiglbock2013model}, \cite{beiglbock2017complete}, \cite{bouchard2015arbitrage}, \cite{BurzoniFritelliMaggisAAP}, \cite{cheridito2020martingale}, \cite{cheridito2017duality}, and \cite{dolinsky2014martingale}, to name but a few. In model-independent approaches, the dual model-class consists of all martingale models (with undefined dynamics of the underlying stochastic process) that are consistent with prices of call options, whereas the trading strategies that super-replicate a payoff pointwise are semi-static.

In this sense, the model-free duality result from \cite{beiglbock2013model} reveals that the dynamic trading position in the underlying security can be considered as a natural counterpart of the martingale property of measures, whereas the static positions in European option corresponds to information on marginal distributions. We will describe the dual counterpart of a dynamic trading position in European options through a martingale property for the prices of these options, i.e., the model-class that is considered for the maximization of the payoff consists of call option-calibrated martingale measures under which the prices of the traded European options are also martingales. 

We study extensively the consequences of the modified model-independent setting for upper bounds of prices for exotic derivatives that emerge as minimal prices among super-replication strategies involving European options and simultaneously as maximal prices over the above described class of financial models.

The remainder of the paper is as follows. Section~\ref{sec_setup_main_results} introduces the setting and provides the main results.  Section~\ref{sec_exa_num} provides several numerical examples. Section~\ref{sec_proofs} contains all the mathematical proofs. Moreover, in Appendix~\ref{sec_extension} we provide extensions of the presented approach to frictions, multiple securities and other dynamically traded options.
\section{Setup and main Results}\label{sec_setup_main_results}
\subsection{Setup}
We consider at $t_0 = 0$ a frictionless discrete-time financial market with a fixed amount of $n \in \N$ times $t_1,\dots,t_n$ and one underlying asset $S=(S_{t_i})_{i=0,1,\dots,n}$. Extensions of this setup are discussed in Appendix~\ref{sec_extension}.

In the classical setup for model-independent pricing, which is referred to as the martingale optimal transport (MOT) case (introduced in \cite{beiglbock2013model}), prices for call options with all strikes and all maturities $t_j$, $j=1,\dots,n$, are observable at initial time $t_0$ and available for static trading. This means that one is able to initiate a buy-and-hold strategy into these call options. Since in this situation every twice-differentiable European payoff $u_i(S_{t_i})$ with $u_i \in C^2(\R)$ can be replicated using call options with different strikes (compare \cite{carr2001towards}), it is natural to also allow initiating static investments in European options.

Moreover, one allows to initiate a trading strategy into the underlying security that is dynamically adjusted over time.
Dynamic trading of European options is however not considered in the MOT setting. 
In contrast, in this paper, we consider dynamic trading in a finite amount of European options. We assume that for each maturity $t_j$ the market offers $N \in \N$ European options\footnote{In practice there may be a different amount of options available for each maturity $t_j$. We then would, purely for technically reasons, additionally consider options with constant payoff $0$ to be able to consider an equal amount of tradable options among maturities.} possessing this expiration date available for trading at all times $t_i < t_j$. We denote the set of options available for dynamic trading by $V$, with $\# V = n \cdot N$. The reduction to the finite subset $V$ accounts for a possible lack of liquidity in European options over time, see also \cite{neufeld2020model}. Additionally, we discuss in detail the case with infinitely many traded options in Section~\ref{sec_n_infinity}.

We denote by $\operatorname{P}_{t_i}\left(v_{j,k}\right)$ the price at time $t_i$ for a European option $v_{j,k} \in V$ with a non-negative Borel-measurable payoff function $v_{j,k}:\R_+\rightarrow \R_+$, where the index $j$ refers to the maturity $t_j$ and $k \in \{1,\dots,N\}$ labels the options. As an underlying sample space we consider $\Omega := \R_+^{n}\times\R_+^{n^2\cdot N}$. We write $\omega=(s_1,\dots,s_n,p_{1,1,1},\dots,p_{n,n,N})=(s,p)$ for $\omega \in \Omega$ with $s \in \R_+^n$ and $p \in \R_+^{n^2\cdot N}$. Then $(S_{t_i})_{i=1,\dots,n}$ is assumed to be the canonical process on the first $n$ components, i.e., for $i =1,\dots,n$ and all $(s,p) \in \Omega$ we have
\[
S_{t_i}(s,p)=s_i,
\]
and where we set $S_0 = s_0$ for some fixed $s_0 \in \R_+$.
Moreover, we set for  $i,j \in \{1,\dots,n\}$, $k \in \{1,\dots,N\}$ 
\[
\operatorname{P}_{t_i}\left(v_{j,k}\right)(s,p)= p_{i,j,k}.
\]
Furthermore, we denote by $\mathcal{P}(\Omega)$ the set of all probability measures on $\Omega$ and we define the filtration $\mathbb{F} = (\mathcal{F}_{t_s})_{s =0,1,\dots,n}$ through
$$
\mathcal{F}_{t_s}:= \sigma\left(\{S_{t_i}\text{ for }  0 \leq i \leq s\}\right)
$$
with $\mathcal{F}_{t_0}$ being trivial.
Next, we allow the financial agent to restrict to price paths she considers admissible. Thus, we introduce, similar as in \cite{bartl2020pathwise}, a subset $ \Xi \subseteq \Omega $ of admissible price paths. Below, in Remark~\ref{rem_xi}, we discuss the role of $\Xi$ in our setting in different examples.
\newpage
\begin{rem}[Choices of $\Xi$]\label{rem_xi}~{\newline}
\begin{itemize}
\item[(a)] It is of course possible to set $\Xi: = \Omega$ and hence not to impose any restrictions on the set of future paths.
\item[(b)] According to \cite{dolinsky2018super} and \cite{neufeld2018buy}, the maximal super-replication price at time $t_0$ for European options in fully incomplete markets coincides with the buy-and-hold-super-replication price and is given by today's value of its concave envelope, i.e., the smallest concave function larger or equal than the payoff function $f$ itself, here denoted by $f ^{\operatorname{conc}}$. For the sub-replication price one correspondingly considers the convex envelope, $f^{\operatorname{conv}}$, which is the greatest convex function smaller than the payoff. This means, we obtain a pricing rule of the form
\[
\operatorname{P}_{t_i}(v_{j,k})\in \left[v_{j,k}^{\operatorname{conv}}(S_{t_i}),v_{j,k}^{\operatorname{conc}}(S_{t_i})\right],
\]
i.e.,
\[
\Xi = \left\{(s,p) \in \Omega~\middle|~p_{i,j,k} \in \left[v_{j,k}^{\operatorname{conv}}(s_i),v_{j,k}^{\operatorname{conc}}(s_i)\right] \text{ for all } i,j,k\right\}.
\]
Using this approach, we also rediscover the standard no-arbitrage bounds for call options (see also \cite{cvitanic1999closed}, \cite{eberlein1997range}, and \cite{frey1999bounds})
\[
\operatorname{P}_{t_i}((S_{t_j}-K)^+)\in 
\left[(S_{t_i}-K)^+,S_{t_i}\right].
\]
\item[(c)] Suppose an investor believes that it is accurate to price call options under the risk-neutral measure of a Black--Scholes model with volatility $\widehat{\sigma}$. To incorporate uncertainty w.r.t. the choice of the volatility ${\sigma}$ one allows for $\sigma \in [\widehat{\sigma}-\varepsilon,\widehat{\sigma}+\varepsilon]$ for some $\varepsilon>0$ such that $\widehat{\sigma}-\varepsilon>0$. Due to the positive vega, pricing of call options in a Black--Scholes model is monotone w.r.t. the choice of $\sigma$ and we obtain a pricing rule of the form
\begin{equation}\label{eq_bs_pricing}
\begin{aligned}
\operatorname{P}_{t_i}\left((S_{t_j}-K)^+\right)\in \bigg[S_{t_i}&\mathcal{N}\left(d_{1,\widehat{\sigma}-\varepsilon}(S_{t_i},K)\right)-K\mathcal{N}\left(d_{2,\widehat{\sigma}-\varepsilon}(S_{t_i},K)\right),\\
&S_{t_i}\mathcal{N}\left(d_{1,\widehat{\sigma}+\varepsilon}(S_{t_i},K)\right)-K\mathcal{N}\left(d_{2,\widehat{\sigma}+\varepsilon}(S_{t_i},K)\right)\bigg]
\end{aligned}
\end{equation}
for $\mathcal{N}(\cdot)$ describing the cumulative distribution function of the standard normal distribution and with 
\begin{equation}\label{eq_d1_bs}
d_{1,\sigma}(x,K) =\frac{\ln\left(x/K\right)+\frac{\sigma^2}{2}(t_j-t_i)}{\sigma \sqrt{t_j-t_i}},\quad d_{2,\sigma}(x,K)=d_{1,\sigma}(x,K)-\sigma\sqrt{t_j-t_i}.
\end{equation}
More precisely, let $(K_{j,k})_{j=1,\dots,n \atop k=1,\dots,N}\subset \R_+$ denote the strikes of the traded call options. Then, we have
\[
V=\left\{v_{j,k}~\middle|~v_{j,k}:x \mapsto \left(x-K_{j,k}\right)^+\text{ for } j=1,\dots,n, k=1,\dots,N\right\}.
\]
and 
\begin{align*}
\Xi=\bigg\{(s,p)\in \Omega~\bigg|~p_{i,j,k} \in  \bigg[s_i&\mathcal{N}\left(d_{1,\widehat{\sigma}-\varepsilon}(s_i,K_{j,k})\right)-K_{j,k}\mathcal{N}\left(d_{2,\widehat{\sigma}-\varepsilon}(s_i,K_{j,k})\right),\\
&s_i\mathcal{N}\left(d_{1,\widehat{\sigma}+\varepsilon}(s_i,K_{j,k})\right)-K_{j,k}\mathcal{N}\left(d_{2,\widehat{\sigma}+\varepsilon}(s_i,K_{j,k})\right)\bigg] \text{ for all } i,j,k\bigg\}.
\end{align*}
This approach can, in principle, be extended to any kind of parametric model. In particular, as shown above, when prices of convex payoffs are increasing w.r.t.\,the input parameter, then the price bounds are attained by the bounds of the interval.
\item[(d)] The pricing rule can also be robust in the sense that it reflects general properties of the market or an admissible underlying process. For example, a Markov property for the valuation of options (similar as in \cite{sester2020robust}) can be incorporated through
\begin{equation}\label{eq_markov}
\begin{aligned}
\Xi = \big\{(s,p) \in \Omega ~|~&\text{For all } i,j,k \text{ there exists some Borel-measurable function } \\ &\hspace{3cm} f_{i,j,k}:\R \rightarrow \R \text{ such that }p_{i,j,k}=f_{i,j,k}(s_i)\big\}.
\end{aligned}
\end{equation}
Or if the difference $t_{i+1}-t_i$ is constant for all $i=1,\dots,n-1$, then a homogeneity assumption similar to \cite{eckstein2020martingale} can be modelled through
\begin{equation}
\begin{aligned}\label{eq_homogen}
\Xi = \big\{(s,p) \in \Omega ~|~&p_{i,j,k}=p_{i+l,j,+l,k} \\
&\text{for all } i,j,k,l \text{ s.t. } 1 \leq i+ l,j+l \leq n, k=1,\dots,N\big\}.
\end{aligned}
\end{equation}
If a discrete time financial market is considered as a discretized version of a continuous time market model, then in the continuous time model the restrictions to paths of the form \eqref{eq_markov} correspond to the requirement that prices of measurable payoffs $f:\R\rightarrow \R$ fulfill
\[
\E_\Q[f(S_t)|~\mathcal{F}_s]=\E_\Q[f(S_t)|~S_s]\text{ for all } 0\leq s\leq t,
\]
where here $(\mathcal{F}_s)_{s \geq 0}$ corresponds to a continuous time filtration and $\Q$ is some martingale measure. In contrast, the analogue of \eqref{eq_homogen} in a continuous time setting is the requirement that prices of derivatives with payoff $f$ fulfill
\[
\E_\Q[f(S_t)|~\mathcal{F}_s]=\E_\Q[f(S_{t+\tau})|~\mathcal{F}_{s+\tau}]\text{ for all } 0\leq s\leq t,\text{ and all } \tau \geq 0.
\]
\item[(e)] Given some $\Xi \subseteq \Omega$ we define for $m\in \N$ $s_i^1,\dots,s_i^m,p_{i,j,k}^1,\dots,p_{i,j,k}^m \in \R_+$
\[
\Xi_{\operatorname{grid}}^m:=\left\{(s,p)\in \Xi~\middle|~s_i \in \{s_i^1,\dots,s_i^m\},~p_{i,j,k}\in \{p_{i,j,k}^1,\dots,p_{i,j,k}^m\}\text{ for all } i,j,k\right\}.
\]
This allows to consider the valuation problem on a discrete grid, which is particularly useful for implementing the approach numerically, i.e., via linear programming, compare also Algorithm~\ref{algo_mot_lp}.
\end{itemize}
\end{rem}

\subsection{Valuation of Derivatives}
We are interested in finding model-free price bounds for some exotic financial derivative $\Phi(S_{t_1},\dots,S_{t_n})$, where $\Phi:\R^n_+ \rightarrow \R_+$ is Borel-measurable. For notational simplicity, we focus on finding the upper bound\footnote{This is no restriction, since with the same approach one can easily obtain the lower bound through the relation $\inf_x f(x)= -\sup_x -f(x)$.}.
\subsubsection{The primal approach}
A first approach to determine the value of $\Phi$ is to compute the expectation of $\Phi$ under a risk-neutral pricing measure associated to a potential model of an underlying financial market, i.e., among all measures restricted to paths $\Xi$ that are arbitrage-free. The maximal model-price determines the upper price bound. We call this approach the primal approach.

 To determine the set of such models, we observe that under any admissible pricing measure $\Q$, the price process $\left(\operatorname{P}_{t_i}\left(v_{j,k}\right)\right)_{i=1,\dots,n}$ is required to be a martingale, from which we obtain the required representation that
\begin{equation}\label{eq_pti_martingale}
\begin{aligned}
\operatorname{P}_{t_i}\left(v_{j,k}\right)=\E_\Q\left[\operatorname{P}_{t_j}\left(v_{j,k}\right)\middle|\mathcal{F}_{t_i}\right]=\E_\Q\left[v_{j,k}(S_{t_j})\middle|\mathcal{F}_{t_i}\right]~\Q\text{-a.s.}
\end{aligned}
\end{equation}
for all $1\leq i\leq j \leq n$, $k=1,\dots,N$. We further, extend the validity of \eqref{eq_pti_martingale} to all $1 \leq i,j \leq n$, since this definition implies $\operatorname{P}_{t_i}\left(v_{j,k}\right)=v_{j,k}~\Q\text{-a.s.}$ if $j \leq i$, i.e., the price processes are assumed to be constant after expiration date (for a fixed price path). 

Additionally, we only consider such models consistent with market prices of call options\footnote{A model with associated pricing measure $\Q$ is said to be consistent with the market price $\pi(v)$ of the option $v$ if $\E_\Q[v]=\pi(v)$.}. According to \cite{breeden1978prices}, consistency of pricing measures w.r.t.\ call option prices for all maturities $t_1,\dots,t_n$ and for all strikes determines the one-dimensional marginal distributions of $S_{t_i}$ for all $i=1,\dots,n$, which we from now on denote by $\mu_i$, i.e., $S_{t_i} \sim \mu_i$. To ensure absence of model-independent arbitrage through static option trading, we assume that $\mu_1 \preceq \mu_2 \preceq \dots \preceq \mu_n$, where $\preceq$ denotes the usual convex order for measures with finite first moments, and that all marginals possess the same mean given by the inital value  $s_0$, compare also \cite{acciaio2016model}, \cite{kellerer1972markov} and \cite{strassen1965existence}. Thus, we include this condition in the following standing assumption.\footnote{From now on, we always assume that Assumption \ref{asu_standing_1} holds and we do not repeat it in the statements of our results.} Further, we require the set $\Xi$ to be Borel-measurable in order to be able to define reasonable integral expressions w.r.t.\ $\Xi$. 

\begin{asu}[Standing Assumption]\label{asu_standing_1}
~
\begin{itemize}
\item[(a)]
The marginals $\mu_1,\dots, \mu_n \in \mathcal{P}(\R_+)$ have finite first moments, mean equal to $s_0\in \R_+$, and 
$$
\mu_1 \preceq \mu_2 \preceq \cdots \preceq \mu_n.
$$
\item[(b)]
The set $\Xi \subseteq\Omega$ is Borel-measurable.
\end{itemize}

\end{asu}

In any potential arbitrage-free model of a financial market, $S$ and  $\operatorname{P}$ are martingales and the marginals of $S$ are fixed by $\mu_1,\dots,\mu_n$ through the required consistency with the observed vanilla option prices. We recall that $V=\{v_{j,k}, j=1,\dots,n, k =1,\dots,N\}$, then a risk-neutral measure of an admissible model is consequently an element of
\begin{align*}
\mathcal{M}_V(\Xi,\mu_1,\dots,\mu_n):=\bigg\{\Q \in \mathcal{P}(\Omega)~\bigg|~ &\Q(\Xi)=1;\\
&\E_\Q[S_{t_{i}}|
\mathcal{F}_{t_{i-1}}]=S_{t_{i-1}} ~\Q\text{-a.s.}\text{ for all } i =1,\dots,n; \\
&\Q \circ S_{t_i}^{-1} = \mu_i \text{ for all } i =1,\dots,n; \\
 &\hspace{0cm} \E_{\Q}[v_{j,k}(S_{t_j})|\mathcal{F}_{t_i}] 
=\operatorname{P}_{t_i}(v_{j,k}) ~\Q\text{-a.s.}\text{ for all } \\ 
&i=1,\dots,n,~j =1,\dots,n,~k=1,\dots,N \bigg\}.
\end{align*}

\begin{rem}
For all $\Q\in \mathcal{M}_V(\Xi,\mu_1,\dots,\mu_n)$ and for all $i,j=1,\dots,n$ the random variable $\operatorname{P}_{t_i}(v_{j,k})$ is $\mathcal{F}_{t_i}^{\Q}$-measurable by \eqref{eq_pti_martingale}, where $\mathcal{F}_{t_i}^{\Q}$ denotes the $\Q$-$\mathcal{F}$-completion of $\mathcal{F}_{t_i}$.
In particular, $\left(\operatorname{P}_{t_i}(v_{j,k})\right)_{i=1,\dots,n}$ is adapted to $\mathbb{F}^{\Q}=\left(\mathcal{F}_{t_i}^{\Q}\right)_{i=1,\dots,n}$. Moreover, note that the martingale property in the definition of $\mathcal{M}_V(\Xi,\mu_1,\dots,\mu_n)$ does not change if we define it w.r.t. $\mathbb{F}^{\Q}$ instead of $\mathbb{F}$.
\end{rem}
The upper price bound of $\Phi$ using the primal approach is then given by the maximal expectation of $\Phi$ w.r.t. measures from ${\mathcal{M}}_V(\Xi,\mu_1,\dots,\mu_n)$, namely
\begin{equation}\label{eq_p_xi_def}
\operatorname{P}_\Xi(\Phi):=\sup_{\Q \in {\mathcal{M}}_V(\Xi,\mu_1,\dots,\mu_n)}\int_{\Omega} \Phi(s)\D \Q(s,p).
\end{equation}
The four properties for a measure $\Q$ to be in $\mathcal{M}_V(\Xi,\mu_1,\dots,\mu_n)$ -- i.e., the paths of $S$ are restricted to $\Xi$, $S$ is a $\Q$-martingale, $S$ possesses correct fixed marginals, and that $\operatorname{P}_{t_i}(v_{j,k})$ can be regarded as a conditional expectation of $v_{j,k}$ -- can also be characterized by integral equations. Thus, the set $\mathcal{M}_V(\Xi,\mu_1,\dots,\mu_n)$ can equivalently be written as
\begin{equation}\label{eq_martingale_char_M_V}
\begin{aligned}
\mathcal{M}_V(\Xi,\mu_1,\dots,\mu_n)=\bigg\{\Q \in \mathcal{P}(\Omega)\bigg| &\int_{\Omega}\one_{\Xi}(s,p)\D\Q(s,p)=1;\\
&\int_{\Omega} H(s_1,\dots,s_i)(s_{i+1}-s_i)\D\Q(s,p)=0 \\
&\hspace{1.5cm}\text{ for all } H \in C_b(\R_+^i),~i =1,\dots,n-1; \\
&\int_{\R_+} u_i(s_i)\D \mu_i(s_i)=\int_{\Omega} u_i(s_i)\D \Q(s,p) \\
&\hspace{1.5cm}\text{ for all } u_i \in C_{\operatorname{lin}}(\R_+,\R_+),~i =1,\dots,n; \\
&\hspace{0cm}\int_{\Omega} H(s_1,\dots,s_i)(v_{j,k}(s_j)-p_{i,j,k})\D\Q(s,p)=0 \\
&\hspace{0cm}\text{ for all }H \in C_b(\R_+^i),~i,j =1,\dots,n,~k=1,\dots,N
\bigg\},
\end{aligned}
\end{equation}
where for any $k \in \N$
$$C_{\operatorname{lin}}\left(\R_+^k,\R_+\right):= \left\{f\in C\left(\R_+^k,\R_+\right)~\middle|~ \sup_{(x_1,\dots,x_k)\in \R_+^k} \frac{f(x_1,\dots,x_k)}{1+\sum_{i=1}^kx_i}  < \infty\right\}
$$
denotes the class of positive continuous functions with linear growth on $\R^k_+$.
In the MOT case Assumption~\ref{asu_standing_1}~(a) ensures non-emptiness of the set
\begin{align*}
\mathcal{M}(\mu_1,\dots,\mu_n):=\bigg\{\Q \in \mathcal{P}(\R_+^n)~\bigg|~&\E_\Q[S_{t_{i+1}}|
\mathcal{F}_{t_i}]=S_{t_i}~\Q\text{-a.s.},\\~&\Q\circ S_{t_i}^{-1} = \mu_i \text{ for all } i =1,\dots,n \bigg\},
\end{align*}
see for example \cite{kellerer1972markov}. We ensure in our setting the non-emptiness of $\mathcal{M}_V(\Xi,\mu_1,\dots,\mu_n)$ by an additional condition, see Theorem~\ref{thm_1}~(a).

\subsubsection{The dual approach}
A second valuation approach relies on the determination of the smallest price among model-independent super-replication strategies of $\Phi$ on $\Xi$. We refer also to \cite{bartl2020pathwise}, \cite{bartl2019duality}, \cite{hou2018robust}, and \cite{mykland2003financial}.
This means we consider strategies of the form 
\begin{equation*}
\begin{aligned}
\Psi^V_{(H_i),(H_{i,j,k}),(u_i)}(s,p):=\sum_{i=1}^n u_i(s_i)&+\sum_{i=1}^{n-1}H_i(s_1,\dots,s_i) (s_{{i+1}}-s_{i})\\
&+\sum_{i=1}^{n-1}\sum_{j=i+1}^{n}\sum_{k=1}^{N} H_{i,j,k}(s_{1},\dots,s_{i})\left(v_{j,k}(s_{j})-p_{i,j,k}\right).
\end{aligned}
\end{equation*}
for functions $u_i:\R_+ \rightarrow \R$, $H_i,H_{i,j,k}:\R_+^i \rightarrow \R$ that can be interpreted as trading positions. In addition to the semi-static trading from the martingale optimal transport approach we encounter $\sum_{i=1}^{n-1}\sum_{j=i+1}^{n}\sum_{k=1}^{N} H_{i,j,k}(s_{1},\dots,s_{i})\left(v_{j,k}(s_{j})-p_{i,j,k}\right)$, which is the profit of a self-financing dynamically adjusted trading position in European options.
 We call this approach the dual approach. Given the marginal distributions $\mu_i$ for $i =1,\dots,n$, the fair price of a strategy $\Psi^V_{(H_i),(H_{i,j,k}),(u_i)}(s,p)$ calculates as $\sum_{i=1}^n \int_{\R} u_i(s_i) \D \mu_i(s_i)$, since the terms 
\[
\sum_{i=1}^{n-1}H_i(s_1,\dots,s_i) (s_{i+1}-s_{i})+\sum_{i=1}^{n-1}\sum_{j=i+1}^{n}\sum_{k=1}^{N} H_{i,j,k}(s_1,\dots,s_i)(v_{j,k}(s_j)-p_{i,j,k})
\]
are profits and losses of self-financing strategies and therefore are considered to be costless.
The upper price bound for $\Phi$ using the dual approach is thus given through the super-replication functional:
\begin{equation}\label{eq_super_hedge}
\begin{aligned}
\operatorname{D}_\Xi(\Phi):=\inf_{\substack{u_i\in C_{\operatorname{lin}}(\R_+,\R_+)\\ H_i,H_{i,j,k} \in C_b(\R^i_+)}} \bigg\{\sum_{i=1}^n \int_{\R_+} u_i(s_i) \D \mu_i(s_i)~\bigg|~&\Psi^V_{(H_i),(H_{i,j,k}),(u_i)}(s,p) \geq \Phi(s)\\ 
&\hspace{2cm}\text{ for all } (s,p) \in \Xi\bigg\}.
\end{aligned}
\end{equation}

\subsection{Main Results}
Our first main result imposes that - under mild conditions - the two presented valuation approaches yield the same value. Furthermore, we state criteria for the non-emptiness and compactness of $\mathcal{M}_V(\Xi,\mu_1,\dots,\mu_n)$, guaranteeing the existence of an optimal pricing measure. We refer to Section~\ref{sec_proofs} for the corresponding proofs of the main results stated in the following theorems and remarks.\\
The following set of continuous functions and set of probability measures on $\Omega$ turn out to be useful for our first main result. Let $$
C_{\operatorname{lin},S} := \left\{f \in C(\Omega) ~\middle|~ \sup_{(s,p)\in \Omega} \frac{|f(s,p)|}{\left(1+\sum_{i=1}^ns_i\right)} < \infty \right\}
$$
and let
$$
\mathcal{P}_{\operatorname{lin},S}:=\left\{\Q \in \mathcal{P}(\Omega)~\middle|~\int_{\Omega} \sum_{i=1}^n s_i \D\Q(s,p) < \infty \right\}.
$$
Moreover, we denote by $\sigma\left(\mathcal{P}_{\operatorname{lin},S},C_{\operatorname{lin},S}\right)$ the weak topology on $\mathcal{P}_{\operatorname{lin},S}$ induced by $C_{\operatorname{lin},S}$.

\begin{thm}\label{thm_1}
Let $v_{j,k} \in C_{\operatorname{lin}}(\R_+,\R_+)$ for all $j=1,\dots,n$, $k=1,\dots,N$. Then, the following holds.
\begin{enumerate}
\item[(a)] The set $\mathcal{M}_V(\Xi,\mu_1,\dots,\mu_n)\subset \mathcal{P}(\Omega)$ is non-empty if and only if there exists some $\Q \in \mathcal{M}(\mu_1,\dots,\mu_n)\subset \mathcal{P}(\R^n_+)$ such that\footnote{By abuse of notation $(S_{t_i})$ denotes in \eqref{eq_condexp_in_p_1} the canonical process on $\R_+^{n}$.}
\begin{equation}\label{eq_condexp_in_p_1}
\begin{aligned}
\left((S_{t_1},\dots,S_{t_n}), \left(\E_\Q[v_{j,k}(S_{t_j})~|~\mathcal{F}_{t_i}]\right)_{i,j =1,\dots,n, \atop k =1,\dots,N}\right) \in \Xi ~\Q\text{-a.s.}
\end{aligned}
\end{equation}
\item[(b)]
Let $\Phi \in C_{\operatorname{lin}}\left(\R_+^n,\R_+\right)$. Further assume that $\mathcal{M}_V(\Xi,\mu_1,\dots,\mu_n) \neq \emptyset$ and that $\Xi$ is closed. Then $\mathcal{M}_V(\Xi,\mu_1,\dots,\mu_n)$ is $\sigma\left(\mathcal{P}_{\operatorname{lin},S},C_{\operatorname{lin},S}\right)$-compact,
\begin{align*}
\operatorname{P}_\Xi(\Phi)=\operatorname{D}_\Xi(\Phi),
\end{align*}
and the primal value in \eqref{eq_p_xi_def} is attained.
\end{enumerate}
\end{thm}

\begin{rem}\label{rem_thm_1}~
\begin{itemize}
\item[(a)]
As a consequence of Theorem~\ref{thm_1}~(a) we have that
\[
\operatorname{P}_\Xi(\Phi)=\sup_{\Q \in \mathcal{M}(\mu_1,\dots,\mu_n):\eqref{eq_condexp_in_p_1}\text{ holds}}\int_{\R^n_+} \Phi(s)\D \Q(s),
\]
i.e., $\operatorname{P}_\Xi(\Phi)$ can be considered as a constrained martingale optimal transport problem.
\item[(b)]
If $\Xi = \Omega$, then, according to Theorem~\ref{thm_1}, the non-emptiness of $\mathcal{M}_V(\Xi,\mu_1,\dots,\mu_n)$ is equivalent to the non-emptiness of $\mathcal{M}(\mu_1,\dots,\mu_n)$, independent of $V$.
\item[(c)] Let $\K \subset \R_+^n$ be compact, and define
\[
\Xi =\left\{(s,p) \in \K \times \R_+^{n^2 \cdot N} \subset \Omega~|~\text{For all }i,j,k \text{ there exists } f_{i,j,k} \in \mathfrak{F}_{i,j,k} \text{ s.t. } p_{i,j,k} =f_{i,j,k}(s_1,\dots,s_i)\right\}
\]
for some classes of functions $\mathfrak{F}_{i,j,k} \subset C_{\operatorname{lin}}\left(\R_+^i,\R_+\right)$ whose restrictions $$
\mathfrak{F}_{i,j,k}|_{\widetilde{\K_i}}:=\left\{f_{i,j,k}|_{\widetilde{\K_i}}:\widetilde{\K_i}\subseteq \R_+^i\rightarrow \R_+~\middle|~ f_{i,j,k} \in \mathfrak{F}_{i,j,k}\right\}
$$ onto any compact set $\widetilde{\K_i} \subset \R^i_+$ are compact in the uniform topology on $C(\widetilde{\K_i})$ and which fulfill
\begin{equation}\label{eq_cond_clin_bounded}
\sup_{f_{i,j,k}\in \mathfrak{F}_{i,j,k}}\sup_{(s_1,\dots,s_i)\in \R_+^i}\frac{f_{i,j,k}(s_1,\dots,s_i)}{1+\sum_{\ell=1}^i s_{\ell}}<\infty.
\end{equation}
Note that, as a consequence of the compactness of $\mathfrak{F}_{i,j,k}|_{\widetilde{\K_i}}$ for any compact set $\widetilde{\K_i}\subset \R_+^i$, the set $\Xi$ is compact. Under these conditions, we define 
\begin{align*}
{\mathcal{Q}}_{(\mathfrak{F}_{i,j,k})}:=\bigg\{\Q \in \mathcal{M}(\mu_1,\dots,\mu_n)~\bigg|~\Q(\K)=1,&\text{ and for all } i,j,k \text{ there exists } f_{i,j,k} \in \mathfrak{F}_{i,j,k}\text{ s.t. } \\
&\hspace{2.5cm}\E_\Q[v_{j,k}(S_{t_j})~|~ \mathcal{F}_{t_i}]=f_{i,j,k} ~\Q\text{-a.s.} \bigg\}
\end{align*}
and
\begin{align*}
\operatorname{D}_{(\mathfrak{F}_{i,j,k})}(\Phi)&:=\inf_{\substack{u_i\in C_{\operatorname{lin}}(\R_+,\R_+)\\ H_i,H_{i,j,k} \in C_b(\R^i_+)}} \bigg\{\sum_{i=1}^n \int_{\R_+} u_i\D \mu_i~\bigg|~\Psi^V_{(H_i),(H_{i,j,k}),(u_i)}(s,\left(f_{i,j,k}(s_1,\dots,s_i)\right)_{i,j,k}) \geq \Phi(s)\\ 
&\hspace{9.3cm}\text{ for all } s \in \K, f_{i,j,k} \in \mathfrak{F}_{i,j,k}\bigg\}.
\end{align*}
Then, under the assumptions of Theorem~\ref{thm_1}~(b), we obtain that
\begin{align*}
\operatorname{P}_\Xi(\Phi)&=\sup_{\Q \in {\mathcal{Q}}_{(\mathfrak{F}_{i,j,k})}}\int_{\R^n_+} \Phi(s) \D \Q(s)=\operatorname{D}_{(\mathfrak{F}_{i,j,k})}(\Phi).
\end{align*}
This particular structure of $\Xi$ is fulfilled for example in the setting discussed in Remark~\ref{rem_xi}~(b) and Remark~\ref{rem_xi}~(c) when the price paths are restricted to $\K$. Indeed, in Remark~\ref{rem_xi}~(b) we obtain for the case of call options with strikes $K_{j,k} \in \R_+$ that 
\begin{equation*}
\mathfrak{F}_{i,j,k}=\left\{(s_1,\dots,s_i)\mapsto (s_i-y)_+\text{ for } y \in [0,K_{j,k}]\right\}.
\end{equation*}
Moreover, in the setting of Remark~\ref{rem_xi}~(c) we get for each $i,j,k$
\begin{equation*}
\mathfrak{F}_{i,j,k}=\left\{(s_1,\dots,s_i)\mapsto s_i\mathcal{N}\left(d_{1,{\sigma}}(s_i,K_{j,k})\right)-K_{j,k}\mathcal{N}\left(d_{2,{\sigma}}(s_i,K_{j,k})\right)\text{ for } \sigma \in [\widehat{\sigma}-\varepsilon,\widehat{\sigma}+\varepsilon]\right\},
\end{equation*}
which satisfies, as required, that $\mathfrak{F}_{i,j,k}|_{\widetilde{\K_i}}\subseteq C(\widetilde{\K_i})$ is compact for every compact set $\widetilde{\K_i}  \subset \R_+^i$ due to the Arzelà–Ascoli theorem. Further, we have
\[
s_i\mathcal{N}\left(d_{1,{\sigma}}(s_i,K_{j,k})\right)-K_{j,k}\mathcal{N}\left(d_{2,{\sigma}}(s_i,K_{j,k} )\right) \leq s_i+K_{j,k},
\] 
i.e., the set is indeed contained in $C_{\operatorname{lin}}\left(\R_+^i,\R_+\right)$.
\end{itemize}
\end{rem}
The following remark shows that without restricting the set of possible prices of the dynamically traded options, i.e., when setting $\Xi: = \Omega$, we do not obtain any improved price bounds in comparison with the classical MOT formulation. This motivates to define a set of restricted pricing rules for European options in order to obtain improved price bounds for $\Phi$.

\begin{rem}\label{rem_eq_mot_with_and_without_constraints}
Let $\Phi \in C_{\operatorname{lin}}\left(\R_+^n,\R_+\right)$. If for all $s \in \R_+^n$ we have that
\begin{equation}\label{eq_sp_in_xi}
(s,\widetilde{p})\in \Xi
\end{equation} 
for $\widetilde{p} \in \R_+^{n^2 \cdot N}$ with $\widetilde{p}_{i,j,k}=v_{j,k}(s_j)$ for all $i,j=1,\dots,n$, $k=1,\dots,N$, then 
\begin{equation}\label{eq_D_P_no_improvement}
\operatorname{D}_\Xi(\Phi)=\sup_{\Q \in \mathcal{M}(\mu_1,\dots,\mu_n)}\int_{\R_+^n}\Phi(s)\D\Q(s).
\end{equation}
Note that condition \eqref{eq_sp_in_xi} holds particularly if $\Xi: = \Omega$.
Further, note that if \eqref{eq_sp_in_xi} holds, then equality \eqref{eq_D_P_no_improvement} holds independently of the amount of traded options $N\in \N$, i.e., gradually increasing the number of dynamically traded options only comes with improved price bounds if we introduce further restrictive pricing rules, as done in the following Theorem~\ref{thm_critical_points}.

\end{rem}

Next, we investigate the sensitivity of the upper price bound of some exotic derivative $\Phi$ w.r.t.\,a change in the pricing rules of dynamically traded options.
More specifically, we study the effect of perturbations of the pricing rules of the dynamically traded options on the upper price bound of $\Phi$ and whether the chosen pricing rule implies improved price bounds in comparison with the bounds that emerge when no option is traded dynamically. To this end, we impose an additional assumption on the shape of the pricing rules.

\begin{asu}\label{asu_form_xi}
Let $\Xi\subset \Omega$ be of the form
\[
\Xi \equiv \Xi_{(\underline{p}_{i,j,k},\overline{p}_{i,j,k})}:=\left\{(s,p)\in \Omega~\middle|~p_{i,j,k} \in \left[\underline{p}_{i,j,k}(s_1,\dots,s_i),\overline{p}_{i,j,k}(s_1,\dots,s_i)\right] \text{ for all } i,j,k\right\}
\]
for some Borel-measurable functions $\underline{p}_{i,j,k},\overline{p}_{i,j,k}:\R_+^i \rightarrow \R_+$ for $i,j=1,\dots,n$, $k=1,\dots,N$.
\end{asu}

\begin{rem}
\begin{itemize}
\item[(a)]A sufficient condition for a set $\Xi \equiv \Xi_{(\underline{p}_{i,j,k},\overline{p}_{i,j,k})}$ satisfying Assumption~\ref{asu_form_xi} to be closed is that each $\underline{p}_{i,j,k},\overline{p}_{i,j,k}:\R_+^i \rightarrow \R_+ $, $i,j=1,\dots,n$, $k=1,\dots,N$ is continuous.
\item[(b)] Note that the examples for $\Xi$ in Remark~\ref{rem_xi}~(b) and in Remark~\ref{rem_xi}~(c) both satisfy Assumption~\ref{asu_form_xi}.
\end{itemize}

\end{rem}

We investigate if the boundaries $[\underline{p}_{i,j,k}(\cdot),\overline{p}_{i,j,k}(\cdot)]$ imply improved price bounds and if not, to which extend the boundaries $[\underline{p}_{i,j,k}(\cdot),\overline{p}_{i,j,k}(\cdot)]$ need to be perturbed to directly imply improved price bounds of $\Phi$ in comparison with the MOT formulation, i.e., to obtain $\operatorname{P}_{\Xi}(\Phi) < \operatorname{P}_{\Omega}(\Phi)$. The following theorem asserts precisely how the pricing rules $\underline{p}_{i,j,k},\overline{p}_{i,j,k}$ have to be defined to obtain improved price bounds for $\Phi$. For this, we define for a fixed financial derivative $\Phi \in C_{\operatorname{lin}}\left(\R_+^n,\R_+\right)$ the set of optimizers of the primal problem
\[
\widehat{\mathcal{M}}_V(\Xi,\mu_1,\dots,\mu_n) := \left\{\Q \in \mathcal{M}_V(\Xi,\mu_1,\dots,\mu_n) ~\text{ s.t. }~\int_{\Xi}\Phi \D \Q = \operatorname{P}_\Xi(\Phi)\right\},
\]
which is non-empty under the assumptions of Theorem~\ref{thm_1}~(b). 
\begin{thm}\label{thm_critical_points}
Let the assumptions of Theorem~\ref{thm_1}~(b) hold and let $\Xi$ be of the form described in Assumption~\ref{asu_form_xi}. Then the following holds.
\begin{itemize}
\item[(a)]
We have that
\[
\operatorname{P}_{\Xi_{(\underline{p}_{i,j,k},\overline{p}_{i,j,k})}}(\Phi)<\operatorname{P}_\Omega(\Phi)
\]
if and only if for all $\Q \in \widehat{\mathcal{M}}_V(\Omega,\mu_1,\dots,\mu_n)$ there exist $i,j \in \{1,\dots,n\},k\in \{1,\dots,N\}$ such that 
\begin{equation}\label{eq_thm_condition1}
\underline{p}_{i,j,k} > \E_\Q\left[v_{j,k}(S_{t_{j}})~\middle|~\mathcal{F}_{t_{i}}\right]  
\end{equation}
or
\begin{equation}\label{eq_thm_condition2}
\overline{p}_{i,j,k} < \E_\Q\left[v_{j,k}(S_{t_{j}})~\middle|~\mathcal{F}_{t_{i}}\right]
\end{equation}
on some Borel-measurable set $A \subset \Omega$ with $\Q(A)>0$.
\item[(b)]
For all $\varepsilon > 0$ such that $\mathcal{M}_{V}\left(\Xi_{(\underline{p}_{i,j,k}+\varepsilon,\overline{p}_{i,j,k})},\mu_1,\dots,\mu_n\right) \neq \emptyset$ we have that
\[
\operatorname{P}_{\Xi_{(\underline{p}_{i,j,k}+\varepsilon,\overline{p}_{i,j,k})}}(\Phi)<\operatorname{P}_{\Xi_{(\underline{p}_{i,j,k},\overline{p}_{i,j,k})}}(\Phi)
\]
if and only if for all $\Q \in \widehat{\mathcal{M}}_{V}\left(\Xi_{(\underline{p}_{i,j,k},\overline{p}_{i,j,k})},\mu_1,\dots,\mu_n\right) $ there exist $i,j \in \{1,\dots,n\},k\in \{1,\dots,N\}$ such that 
\begin{equation}\label{eq_thm_condition3}
\underline{p}_{i,j,k}+\varepsilon > \E_\Q\left[v_{j,k}(S_{t_{j}})~\middle|~\mathcal{F}_{t_{i}}\right] 
\end{equation}
on some Borel-measurable set $A \subset \Omega$ with $\Q(A)>0$.

\item[(c)]
For $\varepsilon > 0$ such that $\mathcal{M}_{V}\left(\Xi_{(\underline{p}_{i,j,k},\overline{p}_{i,j,k}-\varepsilon)},\mu_1,\dots,\mu_n\right) \neq \emptyset$ we have that
\[
\operatorname{P}_{\Xi_{(\underline{p}_{i,j,k},\overline{p}_{i,j,k}-\varepsilon)}}(\Phi)<\operatorname{P}_{\Xi_{(\underline{p}_{i,j,k},\overline{p}_{i,j,k})}}(\Phi)
\]
if and only if  for all $\Q \in \widehat{\mathcal{M}}_{V}\left(\Xi_{(\underline{p}_{i,j,k},\overline{p}_{i,j,k})},\mu_1,\dots,\mu_n\right) $ there exist $i,j \in \{1,\dots,n\},k\in \{1,\dots,N\}$ such that
\begin{equation}\label{eq_thm_condition4}
\overline{p}_{i,j,k}-\varepsilon <  \E_\Q\left[v_{j,k}(S_{t_{j}})~\middle|~\mathcal{F}_{t_{i}}\right] 
\end{equation}
on some Borel-measurable set $A \subset \Omega$ with $\Q(A)>0$.
\end{itemize}
\end{thm}
We remark that, in particular, Theorem~\ref{thm_1}~(b) is applicable to the sets $\Xi_{(\underline{p}_{i,j,k}+\varepsilon,\overline{p}_{i,j,k})}$ and $\Xi_{(\underline{p}_{i,j,k},\overline{p}_{i,j,k}-\varepsilon)}$, respectively. This allows to implement the associated semi-static strategies and to exploit potentially mispriced derivatives. We further highlight that the $\varepsilon$-pertubation in the sets $\Xi_{(\underline{p}_{i,j,k}+\varepsilon,\overline{p}_{i,j,k})}$ and $\Xi_{(\underline{p}_{i,j,k},\overline{p}_{i,j,k}-\varepsilon)}$, respectively, was introduced to study the effect on price bounds, but not to simplify the numerical simulations, as it was done similarly for example in \cite{guo2019computational} to relax the associated martingale constraint.
\subsection{Infinitely many European call options}\label{sec_n_infinity}

If we do not restrict the set of options available for dynamic trading to a fixed \textit{finite} amount of European options, but a priori consider infinitely many call options with  a continuous range of strikes reaching from $0$ to $+\infty$, then, by following the rationale from \cite{carr2001towards}, each positive European payoff can  be replicated by call options and thus every European payoff which only depends on a single value of the underlying security can be considered as being available for dynamic trading. Hence, from now on, we allow for dynamic trading in all options with payoff $v_j(S_{t_j})$ for  $v_j \in \widetilde{V} \subseteq C_{\operatorname{lin}}(\R_+,\R_+)$ and $j \in \{1,\dots,n\}$, where $\widetilde{V}$ is possibly infinite, indexed by $\mathcal{I}_{\widetilde{V}}$.

Similar to Remark~\ref{rem_thm_1} we consider the following formulation of a super-hedging problem. Given a Borel-measurable set $\widetilde{\Xi} \subset \R^n_+$ and some sets $\widetilde{\mathfrak{F}}_{i,j,k}\subset C_{\operatorname{lin}}\left(\R_+^i,\R_+\right)$ for functions $v_{j,k} \in\widetilde{V} \subseteq C_{\operatorname{lin}}(\R_+,\R_+)$, $i,j=1,\dots,n$, $k \in \mathcal{I}_{\widetilde{V}}$ we define
\begin{equation}\label{eq_infinite_prob_q}
\begin{aligned}
\widetilde{\operatorname{D}}_{(\widetilde{\mathfrak{F}}_{i,j,k})}(\Phi):=\inf_{\substack{u_i \in C_{\operatorname{lin}}(\R_+,\R_+) \\ H_i,H_{i,j,k} \in C_b(\R_+^i):\\ \eqref{eq_finitely_many_nonzero} \text{ holds}}} \bigg\{ \sum_{i=1}^n \int_{\R_+} u_i(s_i)\D \mu_i(s_i) ~\bigg|~\sum_{i=1}^n u_i(s_i)&+\sum_{i=1}^{n-1}H_i(s_1,\dots,s_i) (s_{{i+1}}-s_{i})\\
&\hspace{-6cm}+\sum_{i=1}^{n-1}\sum_{j=i+1}^{n} \sum_{k \in \mathcal{I}_{\widetilde{V}}} H_{i,j,k}(s_{1},\dots,s_{i})\left(v_{j,k}(s_{j})-f_{i,j,k}(s_1,\dots,s_i)\right) \geq \Phi(s) \\
&\hspace{1cm}\text{for all } s \in \widetilde{\Xi}, f_{i,j,k} \in \widetilde{\mathfrak{F}}_{i,j,k})\bigg\},
\end{aligned}
\end{equation}
where
\begin{equation}\label{eq_finitely_many_nonzero}
\text{for all } i,j:~ H_{i,j,k},~k \in \mathcal{I}_{\widetilde{V}},\text{ are equal to zero, up to finitely many } k.
\end{equation}
This means $\widetilde{\operatorname{D}}_{(\mathfrak{F}_{i,j,k})}$ corresponds to the minimal super-replication price among strategies where dynamic trading in all options $v_{j,k} \in {\widetilde{V}} \subseteq C_{\operatorname{lin}}(\R_+,\R_+)$ is allowed at time $t_i$. The time $t_i$-price of this option is associated to some $f_{i,j,k} \in \widetilde{\mathfrak{F}}_{i,j,k}$ which is unknown for the financial agent. Thus, the considered strategies super-replicate $\Phi$ pointwise on $\widetilde{\Xi}$ and among  all potential prices $f_{i,j,k}$ in $\widetilde{\mathfrak{F}}_{i,j,k}$.

We obtain the following duality result that allows an interpretation of the super-hedging problem as a maximization problem of expected values of $\Phi$ w.r.t. martingale measures $\Q$ s.t. the $\mathcal{F}_{t_i}$-conditional expectations of $v_{j,k}\in C_{\operatorname{lin}}(\R_+,\R_+)$ can be written in terms of some function $f_{i,j,k}$ from $\widetilde{\mathfrak{F}}_{i,j,k}$.
\begin{thm}\label{thm_infinity_N}
Let $\Phi \in C_{\operatorname{lin}}\left(\R_+^n,\R_+\right)$, let $\widetilde{\Xi} \subset \R_+^n$ be compact, and let each $(\widetilde{\mathfrak{F}}_{i,j,k})_{k \in  \mathcal{I}_{\widetilde{V}}, \atop i,j
\in \{1,\dots,n\}}   \subset  C_{\operatorname{lin}}\left(\R_+,\R_+\right)$ satisfy for all compact $\K \subset \R_+^n$ that $\widetilde{\mathfrak{F}}_{i,j,k}|_\K $ is compact in the uniform topology on $C(\K)$ and such that for all $i,j \in \{1,\dots,n\}$, $k \in  \mathcal{I}_{\widetilde{V}}$
\begin{equation}
\label{eq_linear_growth_N_inf}
\sup_{f_{i,j,k}\in \widetilde{\mathfrak{F}}_{i,j,k} \atop {(s_1,\dots,s_i)\in \R_+^i}}\frac{f_{i,j,k}(s_1,\dots,s_i)}{1+\sum_{\ell=1}^is_\ell}<\infty.
\end{equation}
If the set
\begin{align*}
\widetilde{\mathcal{Q}}_{(\widetilde{\mathfrak{F}}_{i,j,k})}=\bigg\{\Q \in \mathcal{M}(\mu_1,\dots,\mu_n):~&\Q\big(\widetilde{\Xi}\big)=1,~\text{ for all } i,j \text{ and all } v_{j,k} \in \widetilde{V} \subseteq C_{\operatorname{lin}}(\R_+,\R_+)\\
&\text{there exists } {f_{i,j,k}} \in \widetilde{\mathfrak{F}}_{i,j,k}\text{ s.t. }\E_\Q[v_{j,k}(S_{t_j})~|~ \mathcal{F}_{t_i}]={f_{i,j,k}} ~\Q\text{-a.s.} \bigg\}
\end{align*}
is non-empty, then
\[
\widetilde{\operatorname{D}}_{(\widetilde{\mathfrak{F}}_{i,j,k})}(\Phi)=\sup_{\Q \in \widetilde{\mathcal{Q}}_{(\widetilde{\mathfrak{F}}_{i,j,k})}} \int_{\R_+^n} \Phi(s)\D \Q(s).
\]
\end{thm}

\begin{rem}\label{rem_examples_fijk}
An example of sets $\widetilde{\mathfrak{F}}_{i,j,k}$ fulfilling the assumptions of Theorem~\ref{thm_infinity_N} includes for given $i,j, \in \{1,\dots,n\}, k \in  \mathcal{I}_{\widetilde{V}}$ the sets
\begin{align*}
\widetilde{\mathfrak{F}}_{i,j,k}=\bigg\{(s_1,\dots,s_i)\mapsto g(s_i)~\bigg|~ &g \text{ being } 1\text{-Lipschitz with }g(0)=0,\text{ and}\\
&g(S_{t_i})=\E_{\Q}[v_{j,k}(S_{t_j})~|~S_{t_i}]~\Q\text{-a.s. }\text{for some } \Q \in \mathcal{M}(\mu_1,\dots,\mu_n)\bigg\}
\end{align*}
of prices following a Markovian pricing rule.
\end{rem}

\begin{rem}\label{rem_reduction_of_function_class}
The case where $\widetilde{V} \subsetneq C_{\operatorname{lin}}\left(\R_+^i,\R_+\right)$ is a strict subset (still possibly infinitely large) accounts for a possible lack in liquidity. Therefore, one possible choice for $\widetilde{V}$ includes all payoffs of call and put options for a predefined range of strikes. If $\widetilde{V}$ only contains a finite number of payoffs, then we rediscover the result discussed in Remark~\ref{rem_thm_1}~(c).
\end{rem}

\section{Examples and Numerics}\label{sec_exa_num}
\subsection{Examples}
In this section we provide several examples.\footnote{All the codes are available under \htmladdnormallink{https://github.com/juliansester/dynamic\_option\_trading}{https://github.com/juliansester/dynamic\_option\_trading}}  In particular, we compare our approach with the conventional martingale transport approach where semi-static hedging without dynamic trading in options is involved. We start with an empirical study indicating how to choose pricing rules for European call options.

%
%

\begin{exa}[Market Implied Marginals from real financial data]\label{exa_sensitivity}
We consider the marginal distributions $\mu_1$ and $\mu_2$ derived from call and put options on the stock of \emph{Apple Inc.} The data was observed at $t_0 = $ $24$ July $2020$ for $S_{t_0}=389.09$. The considered time to maturities are $t_1-t_0 = 84$ days and $t_2-t_0 = 175$ days respectively.
Due to the short maturities we neglect dividend yields as well as interest rates and discretize the resultant marginal distributions on a discrete grid with $20$ supporting values, where the discretization is performed according to the method proposed in \cite{baker2012martingales} and \cite{guo2019computational} to be able to apply the linear programming approach that is described in Algorithm~\ref{algo_mot_lp}.
We allow for dynamic trading in call options with maturity $t_2$\footnote{We only consider dynamic trading in options with maturity $t_2$ as trading in an option with maturity $t_1$ would not induce a proper dynamic trading position, since such positions are implicitly subsumed in the static component $u_1$. } and strikes $K_k$, i.e., $v_{2,k}(S_{t_2})=\left(S_{t_2}-K_k\right)^+$, where $K_1=360,K_2=340,K_3=320$. We set the standard price bounds $\underline{p}_{1,2,k}(S_{t_1}) = (S_{t_1}-K_k)^+$ and $\overline{p}_{1,2,k}(S_{t_1}) = S_{t_1}$ for $k=1,2,3$, see also Remark~\ref{rem_xi}~(b). Now, we compute numerically the quantities $\operatorname{P}_{\Xi_{(\underline{p}_{i,j,k}+\varepsilon,\overline{p}_{i,j,k})}}(\Phi)$ and $\operatorname{P}_{\Xi_{(\underline{p}_{i,j,k},\overline{p}_{i,j,k}-\varepsilon)}}(\Phi)$ for different values of $\varepsilon$ and for different payoff functions $\Phi$. Further we illustrate the differences between considering $V=\{v_{2,1}\},V=\{v_{2,1},v_{2,2}\}$ and $V=\{v_{2,1},v_{2,2},v_{2,3}\}$ respectively, i.e. we study the effect of including more options for dynamic trading.  The results, using Algorithm~\ref{algo_mot_lp}, are depicted in Figure~\ref{fig_improved_bound_p1}, where we observe the following two effects. First, for an increasing level of $\varepsilon$, the intervals $[\underline{p}_{i,j,k}+\varepsilon,\overline{p}_{i,j,k}]$ and $[\underline{p}_{i,j,k},\overline{p}_{i,j,k}-\varepsilon]$ become tighter, therefore the pricing rule is more restrictive which in turn leads to observably smaller upper price bounds. Second, being able to trade in a higher number of dynamically traded options leads to tighter price intervals.

\begin{figure}[h!]
\begin{center}
\includegraphics[scale=0.55]{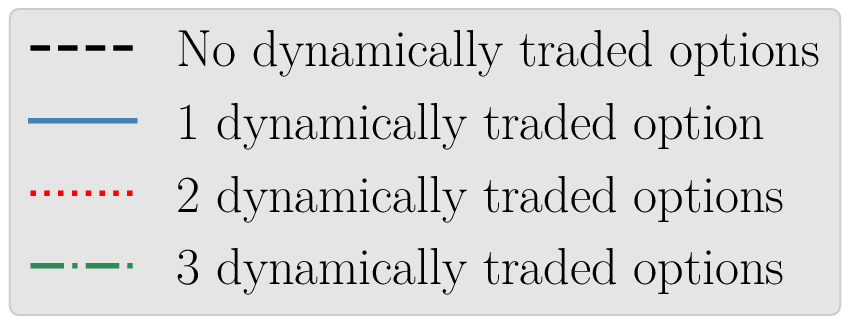}
\includegraphics[scale=0.55]{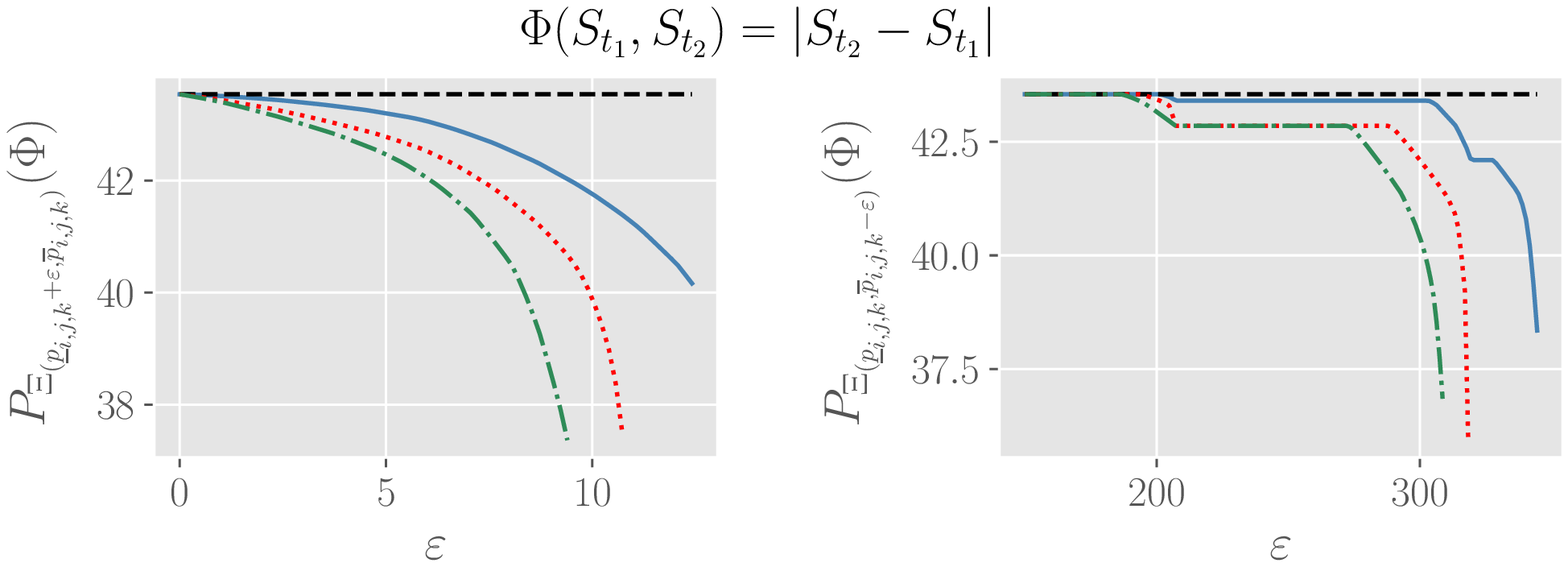}
\includegraphics[scale=0.55]{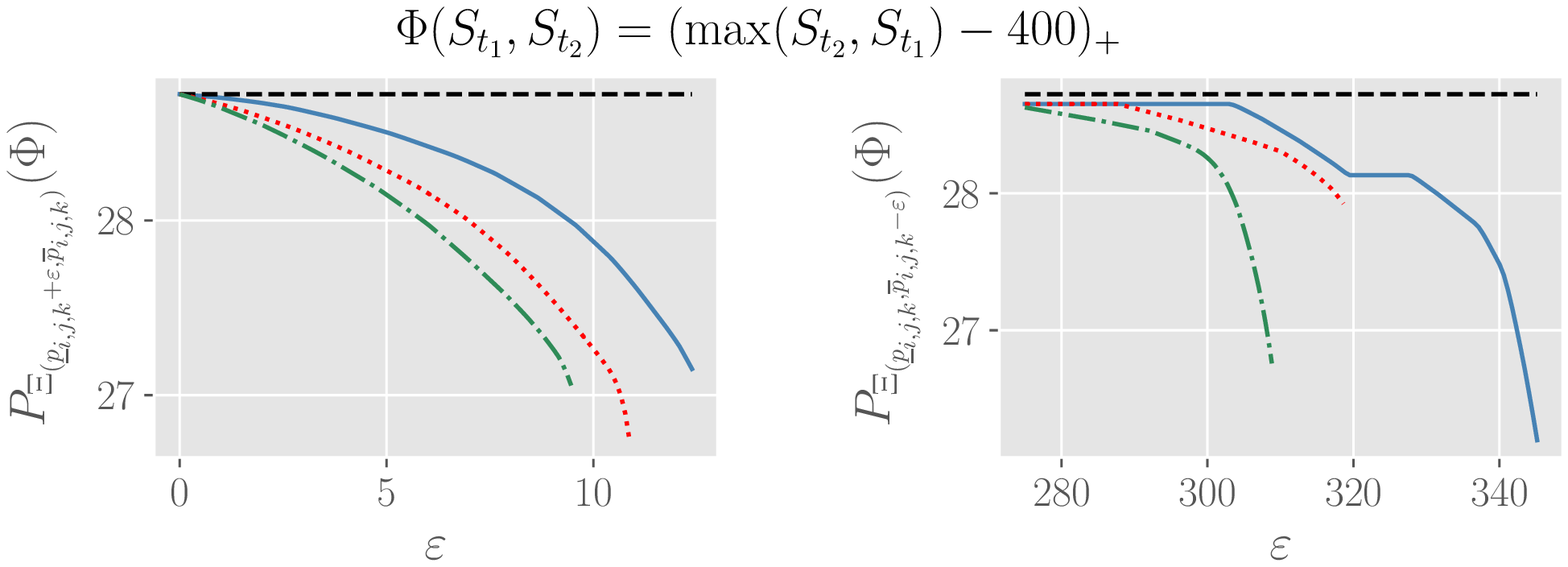}
\includegraphics[scale=0.55]{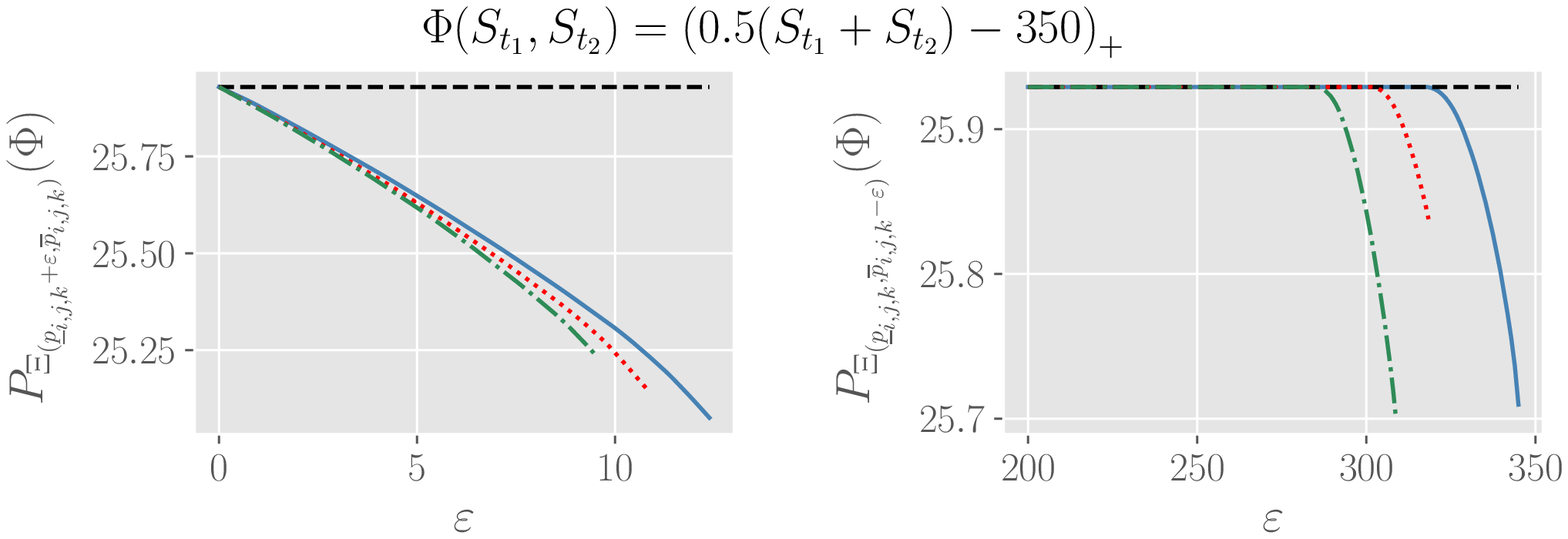}
\includegraphics[scale=0.55]{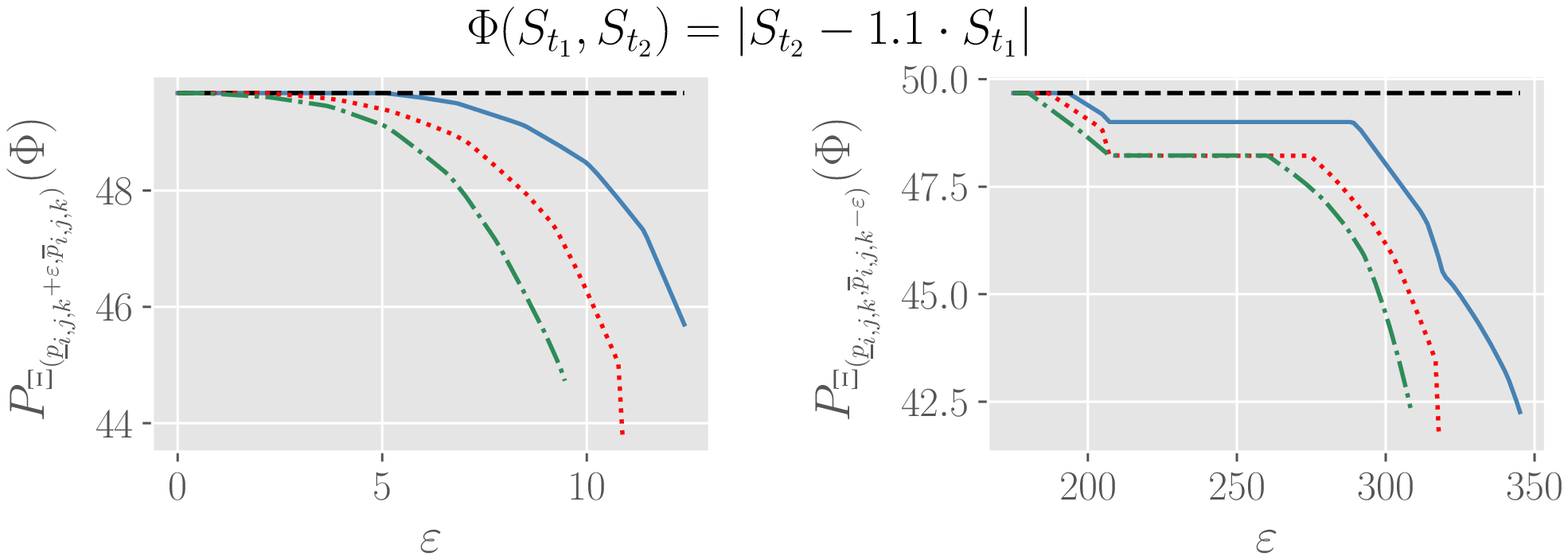}

\caption{The upper price bound for different payoff functions in dependence of a change in the bounds of the pricing rule and in dependence of a different number of considered options for dynamic trading. The price bounds without dynamic option trading (but still with semi-static trading) are indicated by a black dashed line.}\label{fig_improved_bound_p1}
\end{center}
\end{figure}

\end{exa}
\newpage
\begin{exa}[Model-Implied Pricing Rules]
We consider the same market-implied marginals and the same sets $V$ as in Example~\ref{exa_sensitivity}. Then we consider for dynamically traded vanilla options a Black-Scholes type pricing rule of the form \eqref{eq_bs_pricing} and denote for $\varepsilon,\widehat{\sigma} >0$ by 
\[
\Xi_{\widehat{\sigma}-\varepsilon,\widehat{\sigma}+\varepsilon}:=\left\{(s,p) \in \Omega~\middle|~\eqref{eq_bs_pricing}\text{ holds for } \sigma \in [\widehat{\sigma}-\varepsilon,\widehat{\sigma}+\varepsilon] \right\}
\]
the set of admissible paths under a Black-Scholes model with uncertainty in the volatility parameter. We set $\widehat{\sigma} = 0.3$ and depict in Figure~\ref{fig_bs_1} how the robust upper price bounds for several payoff functions $\Phi$ behave under dynamic option trading for a varying number of call options in dependence of $\varepsilon$, using Algorithm~\ref{algo_mot_lp}. We observe in Figure~\ref{fig_bs_1} that the upper price bound becomes smaller for a decreasing level of uncertainty w.r.t.\,the volatility, i.e., for a smaller level of $\varepsilon$. In turn, accounting for more uncertainty through a high level of $\varepsilon$ comes with the drawback of a high upper price bound. Moreover, we observe that the bound can be further decreased through the inclusion of a higher number of traded options.
\begin{figure}[h!]
\begin{center}
\hspace{2cm}\includegraphics[scale=0.55]{eps/legend_real.eps}\newline
\includegraphics[scale=0.55]{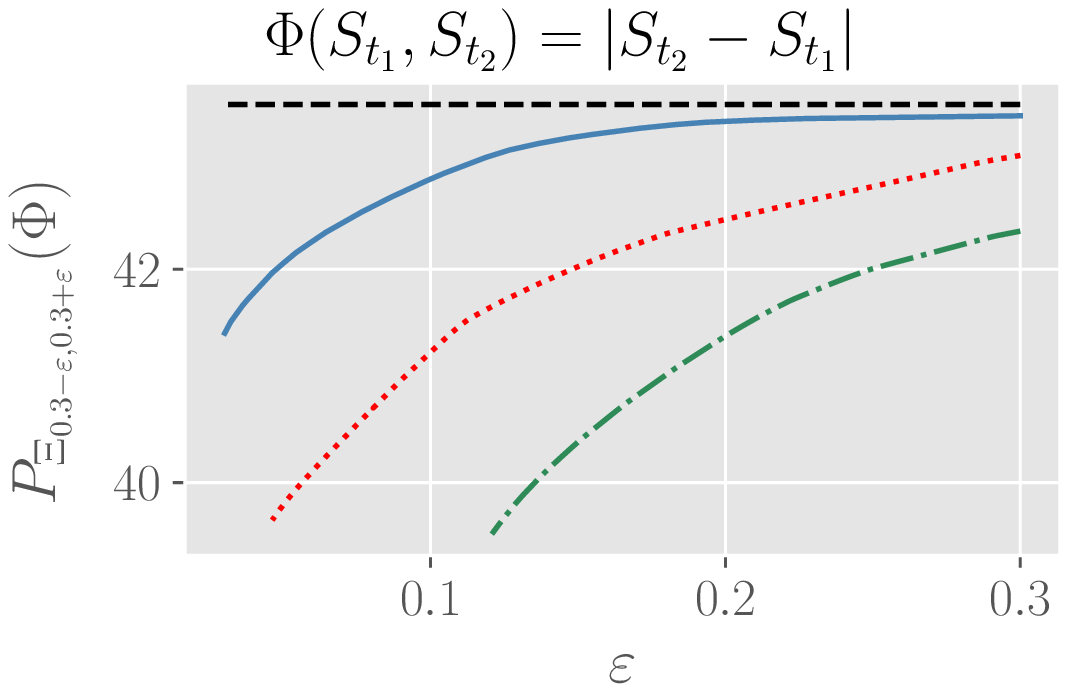}
\includegraphics[scale=0.55]{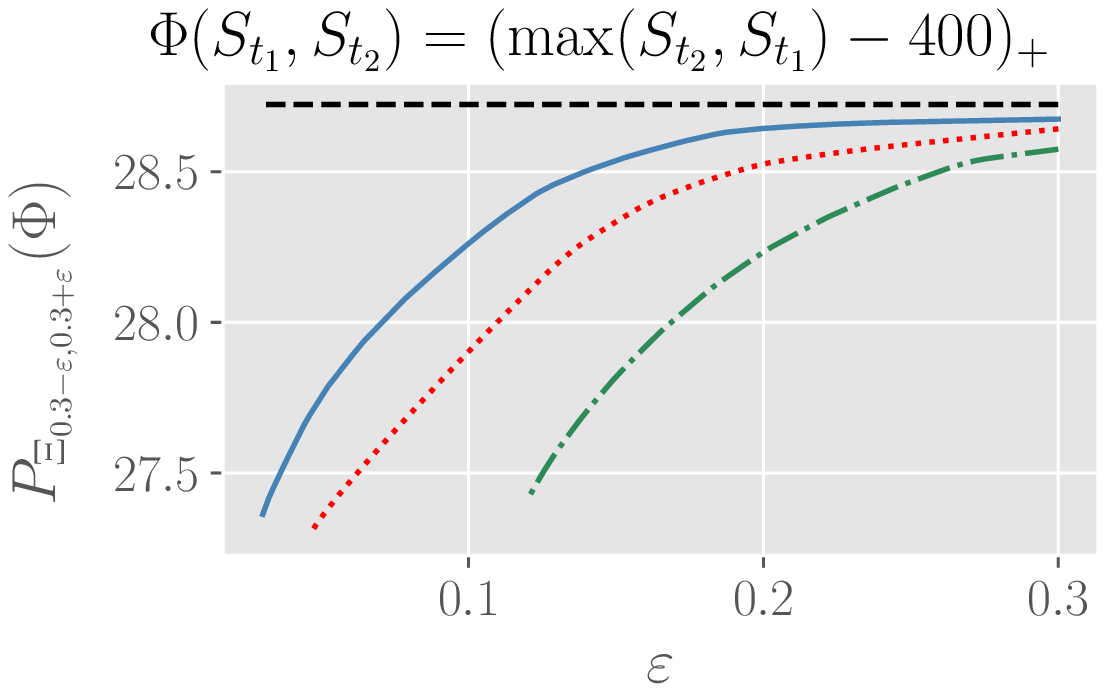}
\includegraphics[scale=0.55]{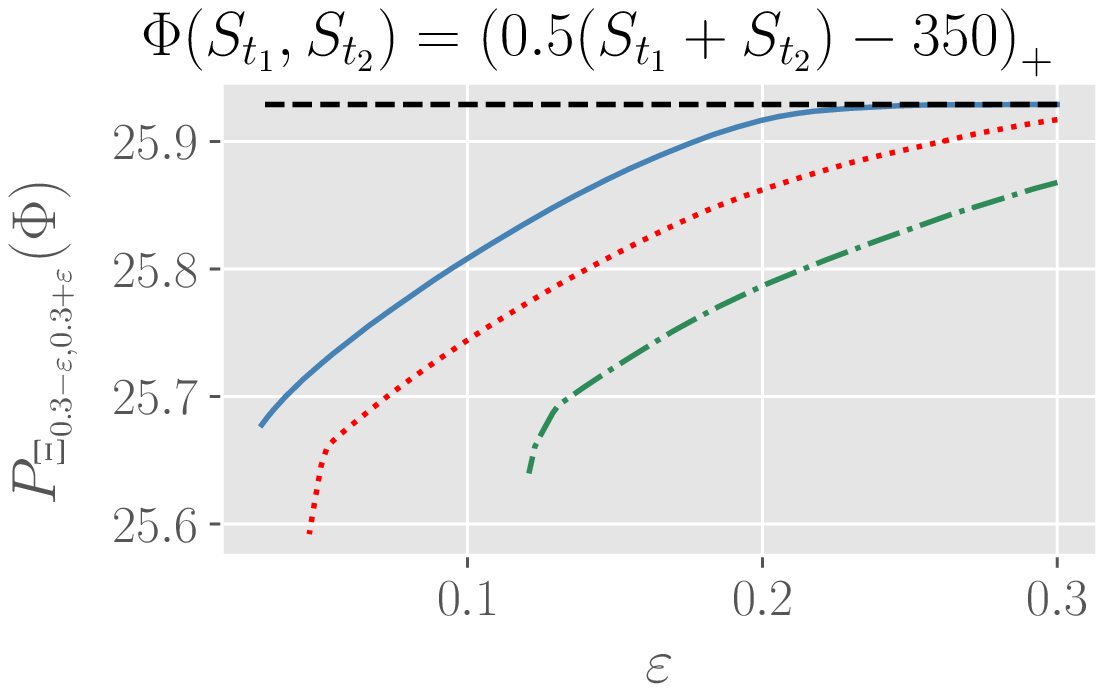}
\includegraphics[scale=0.55]{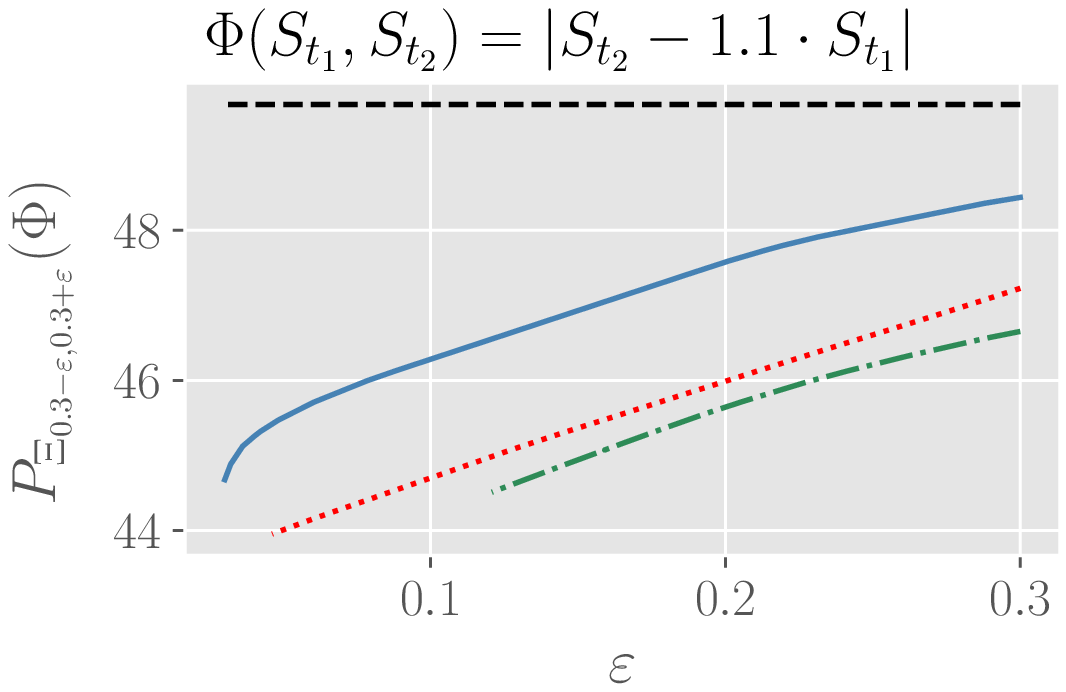}
\caption{The figure shows how the upper robust price bounds behave under a different number of traded options which are priced according to a robust Black-Scholes pricing rule as in \eqref{eq_bs_pricing} in dependence of $\varepsilon$ and for $\widehat{\sigma}=0.3$.}\label{fig_bs_1}
\end{center}
\end{figure}
\end{exa}

\begin{exa}[Three Times, Continuous Marginals]
We consider log-normally distributed marginals
\begin{align*}
&S_{t_i} \sim S_{t_0} \exp\left( \sigma \sqrt{t_i} N_i-\sigma^2 \frac{t_i}{2}\right)\text{ for } i =1,2,3,
\end{align*}
with $S_{t_0}=100$, $\sigma = 10$, $t_i = i$ and $N_i \sim \mathcal{N}(0,1)$ i.i.d. for $i=1,2,3$. The payoff function is an Asian call option of the form $\Phi(S)=\left(\frac{1}{3}\sum_{i=1}^3 S_{t_i}-100\right)^+$.
As dynamically traded options we take into account European call options $v_{2,1}=(S_{t_2}-98)^+$ and $v_{3,1}=(S_{t_3}-98)^+$. We consider as price bounds for the European options $\underline{p}_{1,l,1}(S_{t_l})=(S_{t_l}-98)^+$ and $\overline{p}_{1,l,1}(S_{t_l})=S_{t_l}$ for $l =2,3$ respectively. Then we study, using the neural networks approach which is explained in Section~\ref{sec_penalization}, how the price bounds $\operatorname{P}_{\Xi_{(\underline{p}_{i,j,k}+\varepsilon_1,\overline{p}_{i,j,k}-\varepsilon_2)}}(\Phi)$ behave for increasing $\varepsilon_1,\varepsilon_2$. The results are illustrated in Figure~\ref{fig_2_asset}, where we can observe that increasing $\varepsilon_1,\varepsilon_2$ simultaneously may lead to an even stronger improvement of the price bounds of $\Phi$ in comparison with only increasing either $\varepsilon_1$ or $\varepsilon_2$.
\begin{figure}[h!]
\includegraphics[width=0.7\textwidth]{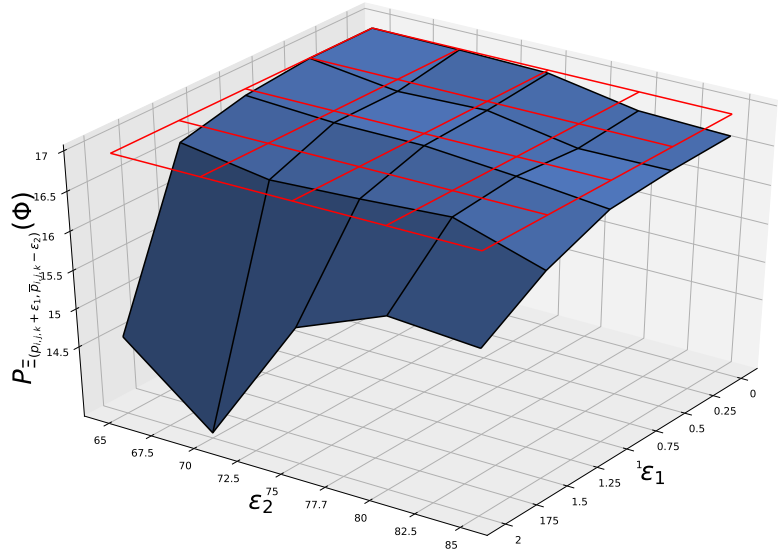}
\caption{The price bounds  $\operatorname{P}_{\Xi_{(\underline{p}_{i,j,k}+\varepsilon_1,\overline{p}_{i,j,k}-\varepsilon_2)}}(\Phi)$ of an Asian option with log-normally distributed marginal distributions and simultaneously increased price bounds of dynamically traded European options.}\label{fig_2_asset}
\end{figure}
\end{exa}

\newpage

\begin{exa}[No arbitrage bounds for call-option prices using $S \& P~500$ data]
We study prices for call options written on constituents of the $S \& P~500$ index at $10$ June $2020$.
In total we investigate $10501$ options. We study to which degree ask and bid prices deviate from the standard no-arbitrage bounds $(S_{t_0}-K)^+$ and $S_{t_0}$ respectively, see also Remark~\ref{rem_xi}~(b). As we do not consider interest rates, we only take into account those options with a short time-to-maturity. Here, we consider only options with time-to-maturity less than $60$ days. The deviation of the average of all normalized prices (in percentage) and of the $5\%$ and $95\%$-quantile of all normalized prices from the no-arbitrage bounds is illustrated in Figure~\ref{fig_histogram}.  We observe a certain amount of options with prices lower than the lower no-arbitrage bound, which can be explained through interest rates and dividend yields.

\begin{figure}[h!]
\includegraphics[scale=0.6]{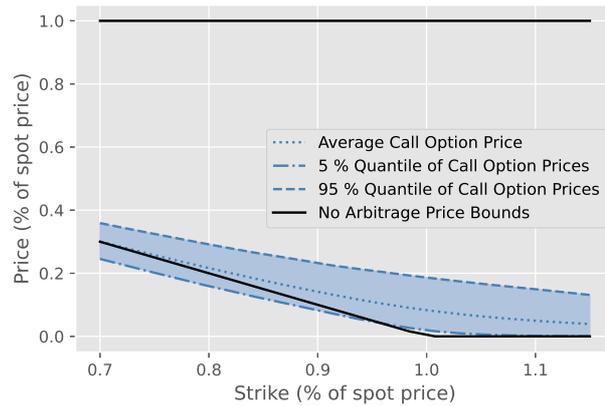}
\caption{The plot shows how prices of call options written on the $S \& P~500
$ deviate from the lower no-arbitrage bound $(S_{t_0}-K)^+$ and the upper no-arbitrage bound~$S_{t_0}$, respectively.}\label{fig_histogram}
\end{figure}

In particular, we realize that the deviation from the upper no-arbitrage bound is much larger than from the lower bound. This is because the payoff functions of call options are convex functions and thus the upper price bound $S_{t_0}$, which is the concave envelope of $(S_{t_j}-K)^+$, is relatively distant from the payoff itself, whereas the convex envelope $(S_{t_0}-K)^+$ is closer to the convex payoff function. For concave payoff functions the situation turns out to be exactly opposite, i.e., the concave envelope is closer to the payoff function than the convex envelope (that appears as a lower bound).
\end{exa}

\subsection{Numerics}\label{sec_numerics}
In this section we sketch and discuss two algorithms to solve the problem of the computation of $\operatorname{D}_\Xi(\Phi)$ and $\operatorname{P}_\Xi(\Phi)$ numerically.
\subsubsection{Linear Programming} 
Given that for all $i=1,\dots,n,~j=1,\dots,n~,k=1,\dots,N$ we have $S_{t_i} \in \{s_1^1,\dots,s_i^{n_i}\}$ and $P_{t_i}(v_{j,k}) \in \{p_{i,j,k}^1,\dots,p_{i,j,k}^{n_{i,j,k}}\}$ - which can always be achieved through a careful discretization of the underlying space\footnote{We remark first that the discretization of continuous marginal distributions needs to be performed such that the discretized marginals keep increasing in convex order, compare \cite{alfonsi2019sampling}, and second, even if the marginals are supported on a discrete grid and do not require a discretization, the price process always needs to be discretized.} - we can formulate a linear program to solve the primal problem $\operatorname{P}_\Xi(\Phi)$ as well as the dual problem $\operatorname{D}_\Xi(\Phi)$.
Here we need to remark that this linear programming approach, however, scales badly with dimensions, but on the contrary yields precise and fast results in low dimensions.

\begin{algorithm}\label{algo_mot_lp}
\SetAlgoLined
\SetKwInOut{Input}{Input}
\SetKwInOut{Output}{Output}
\Input{Marginals $\mu_1,\dots,\mu_n$; Payoff function $\Phi$; Set of dynamically tradable options $V=\{v_{j,k}\}$; Grid $\Xi_{\operatorname{grid}}^m$ as in Remark~\ref{rem_xi}~(e);}
\Output{Minimal $\sum_{i=1}^n \int u_i \D \mu_i$ such that \eqref{ineq_simplex} holds; \\ Minimal $u_i(s), H_i(s_1,\dots,s_i), H_{i,j,k}(s_1,\dots,s_i)$ such that \eqref{ineq_simplex} holds.}
Discretize marginals such that 
$\operatorname{supp}(\mu_i) \subset \Xi_{\operatorname{grid}}^m$, e.g. by the methods from \cite{baker2012martingales} and \cite{guo2019computational};\\
\For{$(s,p)\in \Xi_{\operatorname{grid}}^m$}{Add inequality constraints of the form 
\begin{equation}\label{ineq_simplex}
\begin{aligned}
\sum_{i=1}^n u_i(s)&+\sum_{i=1}^{n-1}H_i(s_1,\dots,s_i) (s_{{i+1}}-s_{i})\\
&+\sum_{i=1}^{n-1}\sum_{j=i+1}^{n}\sum_{k=1}^{N} H_{i,j,k}(s_1,\dots,s_i)\left(v_{j,k}(s_{j})-p_{i,j,k}\right) \geq \Phi(s)
\end{aligned}
\end{equation}}
Minimize 
\[
\sum_{i=1}^n \int u_i \D \mu_i= \sum_{i=1}^n \sum_{(s,p)\in\Xi_{\operatorname{grid}}^m}  u_i(s) \mu_i(\{s\})
\]
w.r.t. $u_i(s), H_i(s_1,\dots,s_i), H_{i,j,k}(s_1,\dots,s_i)$ such that the imposed inequality constraints \eqref{ineq_simplex} are fulfilled. This is possible e.g. via the simplex algorithm, compare \cite{dantzig1998linear}.\\
 \caption{Computation of $\operatorname{D}_\Xi(\Phi)$ via linear programming}
\end{algorithm}
For the computation of $\operatorname{P}_\Xi(\Phi)$, in addition to the the linear programming approach for martingale optimal transport, we obtain supplementary constraints associated to the property
${P}_{t_i}(v_{j,k}) = \E_{\Q}[v_{j,k}(S_{t_j})|\mathcal{F}_{t_i}]$
{ for all } $i=1,\dots,n$

For the computation of  $\operatorname{D}_\Xi(\Phi)$ one obtains for the hedging strategies additional terms of the form $H_{i,j,k}(s_1,\dots,s_i)(v_{j,k}(s_j)-p_{i,j,k})$ that will be considered on the grid induced by the discrete values for $S_{t_i}$ and $P_{t_i}(v_{j,k})$.

 For further details of the approach in the martingale optimal transport setting we refer to the Algorithm~\ref{algo_mot_lp} and to \cite{guo2019computational}, \cite{henry2013automated}. We highlight that Algorithm~\ref{algo_mot_lp} is in line with existing linear programming approaches that are used to solve optimal transport problems. The novelty of the presented algorithm is the adjustment to the extended sample space $\Xi$. We provide the algorithm for sake of completeness.
\subsubsection{Neural networks and penalization}\label{sec_penalization}
We explain how to adjust the approach from \cite{eckstein2019computation} to compute the price bounds involving dynamic option trading.
The adapted algorithm from \cite{eckstein2019computation} is stated in Algorithm~\ref{algo_mot_nn} for the case $p_{i,j,k} \in [\underline{p_{i,j,k}},\overline{p_{i,j,k}}]$ in which one only needs to consider values $p_{i,j,k} \in \{\underline{p_{i,j,k}},\overline{p_{i,j,k}}\}$ since these values lead to the extremal values of the super-hedging strategies. Algorithm~\ref{algo_mot_nn} varies from the approach provided in \cite{eckstein2019computation} by extending the sample space also to the prices of the dynamically traded options.
\begin{algorithm}\label{algo_mot_nn}
\SetAlgoLined
\SetKwInOut{Input}{Input}
\SetKwInOut{Output}{Output}
\Input{Marginals $\mu_1,\dots,\mu_n$; Batch size $B$;  Payoff function $\Phi$; Set of dynamically tradable options $V=\{v_{j,k}\}$; Penalization parameter $\gamma$; Price bound functions $\underline{p_{i,j,k}},\overline{p_{i,j,k}}$ for pricing rules of European options; Number of iterations $N$; Architecture of Neural Networks; Parameters for Adam optimizer;}
\Output{AverageLoss.}
 Initialize neural networks $H_i$, $H_{i,j,k}$, $u_i$, with random weights\;
 $\operatorname{iter} \gets 0$ \\
 \While{$\operatorname{iter}<N$}{
 \For{$b=1:B$}{
 \For{$i=1:n$}{
 Sample $x_i^b \sim \mu_i$\;
 \For{$j=1:n$}{
 \For{$k=1:N$}{
 Sample $p_{i,j,k}^b \sim \mathcal{U}(\{\underline{p_{i,j,k}}(x_1^b,\dots,x_i^b),\overline{p_{i,j,k}}((x_1^b,\dots,x_i^b)\})$\;
 }
 }
 }
 }
 \begin{align*} 
\operatorname{Loss}[\operatorname{iter}] \gets \frac{1}{B} \left(\sum_{p=1}^B\sum_{j=1}^n u_i(x_i^b)\right)\\
+\frac{1}{2}\gamma \frac{1}{B} \max\bigg\{\sum_{p=1}^B \bigg(&\Phi(x_1^b,\cdots,x_n^b)\\
-&\sum_{i=1}^nu_i(x_i^b)-\sum_{i=1}^{n-1}H_i(x_1^b,\dots,x_i^b)(x_{i+1}^b-x_i^b)\\
-&\sum_{i=1}^n\sum_{j=i+1}^n\sum_{k=1}^{N}H_{i,j,k}(x_1^b,\dots,x_i^b)(v_{j,k}(x_j^b)-p_{i,j,k}^b)\bigg),0\bigg\}^{2} ;
\end{align*} 
 Use Adam optimizer to minimize the weights of  $H_i$, $H_{i,j,k}$, $u_i$ w.r.t. $\operatorname{Loss}[\operatorname{iter}]$\;
 $\operatorname{iter} \gets \operatorname{iter} +1$\;
 }
$\operatorname{AverageLoss} \gets \operatorname{Loss}[0.95N : N]$;~\tcp{Average loss over the last $5\%$ of Iterations}
 \caption{Computation of $\operatorname{D}_\Xi(\Phi)$ via penalization}\label{algo_nn}
\end{algorithm}

In contrast to the linear programming approach, this algorithm scales very well with dimensions, i.e., with an increasing number of marginals and of considered underlying securities. However, the choice of the involved hyper-parameters turns out to be a rather complicated task, as it was already observed in \cite{henry2019martingale}. Within our numerical examples we decided to mainly stick to the parameters used in \cite{eckstein2019computation} and \cite{aquino2019bounds} by choosing $\gamma = 10000$, neural networks with $5$ hidden layers, $64\cdot n$ neurons and \emph{ReLu} activation functions. The batch size was $2^{10+n}$ and the optimization was performed by an Adam optimizer with standard parameters for $N= 50000$ iterations. To reduce the variance of the results we finally average over $30$ independent simulations.

\section{Proofs}\label{sec_proofs}
In this section we provide all proofs of the mathematical statements from the previous sections.

\begin{proof}[Proof of Theorem~\ref{thm_1}~(a)]
First, let $\mathcal{M}_V(\Xi,\mu_1,\dots,\mu_n)$ be non-empty. Then pick some measure ${\Q_1} \in \mathcal{M}_V(\Xi,\mu_1,\dots,\mu_n) \subset \mathcal{P}(\Omega)$. We define a measure $\Q_2 \in \mathcal{P}(\R_+^n)$ through 
\[
\Q_2 := {\Q_1} \circ S^{-1}
\]
with $S:\Omega \rightarrow \R_+^n$, $S(s,p)=s$.
The measure $\Q_2$ is contained in $\mathcal{M}(\mu_1,\dots,\mu_n)$ as martingale and marginal properties of $\Q_2$ are inherited from ${\Q_1}$. Moreover,  for all $i,j=1,\dots,n$, $k=1,\dots,N$ and all Borel-measurable sets $A \subset \R_+^i$ we have that 
\begin{align*}
&\int_\Omega \one_A(s_1,\dots,s_i)\E_{{\Q_1}}\left[v_{j,k}(S_{t_j})~\middle|~\mathcal{F}_{t_i}\right](s,p)\D {\Q_1}(s,p)\\
=&\int_\Omega \one_A(s_1,\dots,s_i)v_{j,k}(s_j)\D {\Q_1}(s,p)\\
=&\int_{\R^n_+} \one_A(s_1,\dots,s_i)v_{j,k}(s_j)\D {\Q_2}(s)\\
=&\int_{\R^n_+} \one_A(s_1,\dots,s_i)\E_{\Q_2}\left[v_{j,k}(S_{t_j})~\middle|~\mathcal{F}_{t_i}\right](s)\D {\Q_2}(s)\\
=&\int_{\Omega}\one_A(s_1,\dots,s_i)\E_{{\Q_2}}\left[v_{j,k}(S_{t_j})~\middle|~\mathcal{F}_{t_i}\right]\circ S(s,p)\D {\Q_1}(s,p).
\end{align*}
Thus, we obtain ${\Q_1}$-almost surely that
\[
\E_{{\Q_2}}\left[v_{j,k}(S_{t_j})~\middle|~\mathcal{F}_{t_i}\right]\circ S = \E_{{\Q_1}}\left[v_{j,k}(S_{t_j})~\middle|~\mathcal{F}_{t_i}\right]=\operatorname{P}_{t_i}(v_{j,k}).
\]
This implies, by using the definition of $\Q_2$, that
\begin{align*}
&\Q_2\left(\left\{s\in \R_+^n~\middle|~\left(s,\E_{{\Q_2}}\left[v_{j,k}(S_{t_j})~\middle|~\mathcal{F}_{t_i}\right](s)_{{i,j =1,\dots,n, \atop k =1,\dots,N}}\right) \in \Xi\right\}\right)\\
=&{\Q_1}\left(\left\{(s,p)\in \Omega ~\middle|~\left(S(s,p),\E_{{\Q_2}}\left[v_{j,k}(S_{t_j})~\middle|~\mathcal{F}_{t_i}\right]\circ S(s,p)_{{i,j =1,\dots,n, \atop k =1,\dots,N}}\right)\in \Xi\right\}\right)\\
=&{\Q_1}\left(\left\{(s,p)\in \Omega ~\middle|~\left(S(s,p),\operatorname{P}_{t_i}(v_{j,k})(s,p)_{{i,j =1,\dots,n, \atop k =1,\dots,N}}\right)\in \Xi\right\}\right)={\Q_1}(\Xi)=1,
\end{align*}
and thus \eqref{eq_condexp_in_p_1} is fulfilled.
Conversely, let \eqref{eq_condexp_in_p_1} be valid for some ${\Q_3} \in \mathcal{M}(\mu_1,\dots,\mu_n)\subset \mathcal{P}(\R^n_+)$.
We define a measure ${{\Q_4}}\in \mathcal{P}(\Omega)$ through
\begin{equation}\label{eq_construction_q4}
{{\Q_4}}:={\Q_3} \circ g^{-1}
\end{equation}
for 
\[
g: s \mapsto  \left(s, \left(\E_{{\Q_3}}[v_{j,k}(S_{t_j})~|~\mathcal{F}_{t_i}](s)\right)_{i,j,=1,\dots,n \atop k=1,\dots,N}\right)
\]
Then $\Q_4(\Xi)= {\Q_3}\left(\left\{s \in \R^n_+~\middle|~ g(s) \in \Xi\right\}\right)=1$ is ensured through \eqref{eq_condexp_in_p_1}, and we further have for all $i,j=1,\dots,n$, $k=1,\dots,N$ and $H \in C_b(\R_+^i)$ that
\begin{align*}
&\int_{\Omega} H(s_1,\dots,s_i)(v_{j,k}(s_j)-p_{i,j,k})\D{{\Q_4}}(s,p)\\
=&\int_{\R_+^n} H(s_1,\dots,s_i)(v_{j,k}(s_j)-\E_{{\Q_3}}[v_{j,k}(S_{t_j})~|~\mathcal{F}_{t_i}](s))\D{{\Q_3}}(s)=0.
\end{align*}
Hence $\Q_4 \in \mathcal{M}_V(\Xi,\mu_1,\dots,\mu_n)$, since the martingale and marginal constraints are inherited from $\Q_3$.
\end{proof}

\begin{proof}[Proof of Theorem~\ref{thm_1}~(b)]
We aim at applying the biconjugate duality theorem \cite[Theorem 2.2.]{bartl2019robust} to $\operatorname{D}_\Xi$. A similar proof of a martingale transport duality under additional constraints can be found in \cite{eckstein2020martingale} and \cite{ansari2020improved}. Note that, by abuse of notation,  we have $C_{\operatorname{lin}}\left(\R_+^n,\R_+\right) \subset C_{\operatorname{lin},S}$. First, we extend the domain of the super-replication functional $\operatorname{D}_\Xi(\cdot)$ from payoffs defined only on $\R_+^n$ to payoffs defined on $\Omega$  by considering $\operatorname{D}_\Xi: C_{\operatorname{lin},S} \rightarrow \R$.
Observe that $\operatorname{D}_\Xi(\cdot)$ is convex and increasing on $C_{\operatorname{lin},S}.$ Moreover, the fulfilment of condition (R1) from \cite[Theorem 2.2.]{bartl2019robust} follows analogously as in the proof of \cite[Theorem 3.3.]{eckstein2020martingale}.
An application of \cite[Theorem 2.2.]{bartl2019robust} yields
\begin{equation}\label{eq_convex_duality}
\operatorname{D}_\Xi(\Phi) = \sup_{\Q \in \mathcal{P}_{\operatorname{lin},S}}\left(\int_{\Omega} \Phi(s)\D\Q(s,p)-\operatorname{D}_\Xi^{*}(\Q)\right),
\end{equation}
where the convex conjugate $\operatorname{D}_\Xi^{*}$ of $\operatorname{D}_\Xi$ is defined through 
\[
\operatorname{D}_\Xi^{*}(\Q) = \sup_{f \in C_{\operatorname{lin},S}} \left\{\int_{\Omega} f(s,p)\D\Q(s,p)- \operatorname{D}_\Xi(f)\right\}.
\]
Moreover, by \cite[Theorem 2.2.]{bartl2019robust} we also obtain that all sublevel sets $\left\{\Q \in \mathcal{P}_{\operatorname{lin},S}~\middle|~\operatorname{D}_\Xi^{*}(\Q) \leq c\right\}, c\in \R,$ are $\sigma\left(\mathcal{P}_{\operatorname{lin},S},C_{\operatorname{lin},S}\right)$-compact.
We want to show that
\[
\operatorname{D}_\Xi^{*}(\Q) = \begin{cases}
0 &\text{if } \Q \in \mathcal{M}_{V}(\Xi,\mu_1,\dots,\mu_n), \\
\infty &\text{else.}
\end{cases}
\]
W.l.o.g. assume $\Xi \neq \Omega$, else $\Q(\Xi)=1$ is trivially satisfied.
By Urysohn's Lemma, there exist functions $(f_m)_{m \in \N} \subset C_b(\Omega) \subset C_{\operatorname{lin},S}.$ which are $0$ on $\Xi$ and converge pointwise and monotonically to $\infty \cdot \one_{\Xi^c}$ for $m \rightarrow \infty$. Thus $\operatorname{D}_\Xi(f_m) \leq 0$ for all $m\in \N$ and we obtain 
\begin{equation}\label{eq_dual_infinity}
\operatorname{D}_\Xi^{*}(\Q) \geq \sup_m  \left\{\int_{\Omega} f_m(s,p)\D\Q(s,p)- \operatorname{D}_\Xi(f_m)\right\} \geq \infty \cdot \Q(\Xi^c).
\end{equation}
Therefore $\operatorname{D}_\Xi^{*}(\Q) = \infty$ if $\Q(\Xi^c)>0$.
Assume from now on that $\Q(\Xi)=1$. Next, we compute $\operatorname{D}_\Xi^{*}(\Q)$. We first use the relation $-\inf -f = \sup f$ and obtain
\begin{align*}
\operatorname{D}_\Xi^{*}(\Q) &=\sup_{f \in C_{\operatorname{lin},S}.}\sup_{\substack{u_i\in C_{\operatorname{lin}}(\R_+,\R_+)\\ H_i,H_{i,j,k} \in C_b(\R^i):\\
\Psi^V_{(H_i),(H_{i,j,k}),(u_i)} \geq f \text{ on } \Xi }}\left\{\int_{\Omega} f(s,p)\D\Q(s,p)-\sum_{i=1}^n\int_{\R_+} u_i(s)\D\mu_i(s)\right\}.
\end{align*}
We observe that $\Psi^V_{(H_i),(H_{i,j,k}),(u_i)} \in C_{\operatorname{lin},S}$. Thus we may plug in for $f$ the strategy $\Psi^V_{(H_i),(H_{i,j,k}),(u_i)}$ to get
\begin{align*}
\operatorname{D}_\Xi^{*}(\Q) &= \sup_{u_i \in C_{\operatorname{lin}}(\R_+,\R_+)} \sum_{i=1}^n\left(\int_{\Xi} u_i(s_i)\D\Q(s,p)-\int_{\R_+} u_i(s_i)\D\mu_i(s_i) \right) \\
&+ \sup_{H_i \in C_{b}(\R_+^i)} \sum_{i=1}^{n-1}\left(\int_{\Xi} H_i(s_1,\dots,s_i)(s_{i+1}-s_i)\D\Q(s,p)\right)\\
&+\sup_{H_{i,j,k} \in C_{b}(\R_+^i)} \sum_{i=1}^{n-1}\sum_{j=i+1}^{n}\sum_{k=1}^{N}\left(\int_{\Xi} H_{i,j,k}(s_1,\dots,s_i)(v_{j,k}(s_j)-p_{i,j,k})\D\Q(s,p)\right).
\end{align*}
Then, by the characterization of $\mathcal{M}_{V}(\Xi,\mu_1,\dots,\mu_n)$ through integrals in \eqref{eq_martingale_char_M_V}, and by \eqref{eq_dual_infinity}, we see that
\[
\operatorname{D}_\Xi^{*}(\Q) = \begin{cases}
0 &\text{if } \Q \in \mathcal{M}_{V}(\Xi,\mu_1,\dots,\mu_n), \\
\infty &\text{else }.
\end{cases}
\]
Hence, we conclude from \eqref{eq_convex_duality} that
\begin{align*}
\operatorname{D}_\Xi(\Phi)&=\sup_{\Q \in \mathcal{P}_{\operatorname{lin},S}}\left(\int_{\Omega} \Phi(s)\D\Q(s,p)-\operatorname{D}_\Xi^{*}(\Q)\right)\\
&=\sup_{\Q \in\mathcal{M}_{V}(\Xi,\mu_1,\dots,\mu_n)}\int_{\Omega} \Phi(s)\D\Q(s,p)\\
&=\operatorname{P}_\Xi(\Phi).
\end{align*}
Finally, $\sigma\left(\mathcal{P}_{\operatorname{lin},S},C_{\operatorname{lin},S}\right)$-compactness of $\mathcal{M}_{V}(\Xi,\mu_1,\dots,\mu_n)$ and the attainment of the primal value follows directly from \cite[Theorem 2.2.]{bartl2019robust}, since $\mathcal{M}_{V}(\Xi,\mu_1,\dots,\mu_n)=\left\{\Q \in \mathcal{P}_{\operatorname{lin},S}~\middle|~ \operatorname{D}_\Xi^{*}(\Q) \leq 0\right\}$ is $\sigma\left(\mathcal{P}_{\operatorname{lin},S},C_{\operatorname{lin},S}\right)$-compact	and, by abuse of notation, $\Phi \in C_{\operatorname{lin}}(\R_+^n,\R_+)\subset C_{\operatorname{lin},S}$.
\end{proof}

\begin{proof}[Proof of Remark~\ref{rem_thm_1}~(a)]
W.l.o.g.\ assume that $\mathcal{M}_V(\Xi,\mu_1,\dots,\mu_n)\neq \emptyset$ which by Theorem~\ref{thm_1}~(a) is equivalent to the non-emptiness of $\{\Q \in \mathcal{M}(\mu_1,\dots,\mu_n):\eqref{eq_condexp_in_p_1}\text{ holds}\}$, else the assertion of Remark~\ref{rem_thm_1}~(a) holds trivially. 
Let $\Q_1 \in \mathcal{M}_V(\Xi,\mu_1,\dots,\mu_n)$. Then, according to the proof of Theorem~\ref{thm_1}~(a), there exists some $\Q_2 \in \mathcal{M}(\mu_1,\dots,\mu_n)$ such that \eqref{eq_condexp_in_p_1} holds and such that we further have $\int_\Omega \Phi \D \Q_1=\int_{\R_+^n}\Phi\D\Q_2$. Analogously, for each $\Q_3 \in \mathcal{M}(\mu_1,\dots,\mu_n)$ such that \eqref{eq_condexp_in_p_1} holds we can find some $\Q_4 \in \mathcal{M}_V(\Xi,\mu_1,\dots,\mu_n)$ with $\int_\Omega \Phi \D \Q_4=\int_{\R_+^n}\Phi\D\Q_3$.
\end{proof}

\begin{proof}[Proof of Remark~\ref{rem_thm_1}~(c)]
We first note that the validity of \eqref{eq_condexp_in_p_1} for $\Q \in \mathcal{M}(\mu_1,\dots,\mu_n)$ is equivalent to the fact that $\Q(\K)=1$ and that for all $i,j=1,\dots,n$, $k=1,\dots,N$ there exists some $f_{i,j,k} \in \mathfrak{F}_{i,j,k}$ such that $\E_\Q[v_{j,k}(S_{t_j})~|~S_{t_i},\dots,S_{t_1}]=f_{i,j,k}$ $\Q$-a.s.  According to Remark~\ref{rem_thm_1}~(a), this explains $\operatorname{P}_\Xi(\Phi)=\sup_{\Q \in {\mathcal{Q}}_{(\mathfrak{F}_{i,j,k})}}\int_{\R^n_+} \Phi(s) \D \Q(s)$.

Next, we see that 
\begin{equation}
\Psi^V_{(H_i),(H_{i,j,k}),(u_i)}(s,\left(f_{i,j,k}(s_1,\dots,s_i)\right)_{i,j,k}) \geq \Phi(s)\text{ for all }s \in \K, f_{i,j,k} \in \mathfrak{F}_{i,j,k}
\end{equation}
if and only if 
\begin{align}\label{eq_condition_inf}
\inf_{f_{i,j,k} \in \mathfrak{F}_{i,j,k}}\Psi^V_{(H_i),(H_{i,j,k}),(u_i)}(s,\left(f_{i,j,k}(s_1,\dots,s_i)\right)_{i,j,k}) \geq \Phi(s) \text{ for all }s \in \K
\end{align}
As in the proof of Theorem~\ref{thm_1}~(b), in equation \eqref{eq_convex_duality}, we compute the biconjugate representation of the super-replication functional $\operatorname{D}_{(\mathfrak{F}_{i,j,k})}$. To that end, we first obtain for every measure $\Q \in \mathcal{P}(\R_+^n)$ with finite first moments that its convex conjugate satisfies
\begin{align}\label{eq_biconj_proof_inf}
\operatorname{D}_{(\mathfrak{F}_{i,j,k})}^{*}(\Q)=\sup_{f \in C_{\operatorname{lin}}\left(\R_+^n,\R_+\right)}\sup_{\substack{u_i\in C_{\operatorname{lin}}(\R_+,\R_+)\\ H_i,H_{i,j,k} \in C_b(\R^i)}:\eqref{eq_condition_inf}\text{ holds}}\left\{\int_{\R_+^n} f(s)\D\Q(s)-\sum_{i=1}^n\int_{\R_+} u_i(s)\D\mu_i(s)\right\}.
\end{align}
Analogue to the proof of Theorem~\ref{thm_1}~(b) we aim at showing that 
\begin{equation}\label{eq_cases}
\operatorname{D}_{(\mathfrak{F}_{i,j,k})}^{*}(\Q) =\begin{cases}
0 &\text{if } \Q \in {\mathcal{Q}}_{(\mathfrak{F}_{i,j,k})},\\
\infty &\text{if } \Q \not\in {\mathcal{Q}}_{(\mathfrak{F}_{i,j,k})}.
\end{cases}
\end{equation}
To this end, note first that by the same arguments as in the proof of Theorem~\ref{thm_1}~(b) we obtain that $\operatorname{D}_{(\mathfrak{F}_{i,j,k})}^{*}(\Q)=\infty$ if $\Q(\K^c)>0$. Assume therefore from now on that $\Q(\K)=1$.
Next, we want to show that for all $i,j,k$ we have
\begin{equation}\label{eq_phi_in_c_lin}
\R_+^n \ni s=(s_1,\dots,s_n)\mapsto\inf_{f_{i,j,k} \in \mathfrak{F}_{i,j,k}}\Psi^V_{(H_i),(H_{i,j,k}),(u_i)}(s,\left(f_{i,j,k}(s_1,\dots,s_i)\right)_{i,j,k}) \in C_{\operatorname{lin}}\left(\R_+^n,\R_+\right).
\end{equation}For every $H_{i,j,k} \in C_b(\R^i_+)$ and every $i,j,k$ define 
\begin{align*}
g:\mathfrak{F}_{i,j,k} \times \R^i_+ &\rightarrow \R\\
\left(f_{i,j,k}, (s_1,\dots,s_i)\right)&\mapsto H_{i,j,k}(s_1,\dots,s_i)\left(v_{j,k}(s_j)-f_{i,j,k}(s_1,\dots,s_i)\right),
\end{align*}
and for any compact set $\widetilde{\K_i} \subset \R_+^i$ let $g|_{\widetilde{\K_i}}:\mathfrak{F}_{i,j,k}|_{\widetilde{\K_i}} \times \widetilde{\K_i}  \rightarrow \R$ be the restriction of $g$ onto $\mathfrak{F}_{i,j,k} \times \widetilde{\K_i} $.


Note that $g|_{\widetilde{\K_i}}$ is continuous, as for any  $\left(f_{i,j,k}^{(N)}\right)_{N \in \N}|_{\widetilde{\K_i}} \subset \mathfrak{F}_{i,j,k}|_{\widetilde{\K_i}}$ converging uniformly on $\widetilde{\K_i}$ to some $f_{i,j,k}|_{\widetilde{\K_i}}$ and $\left(s_1^{(N)},\dots,s_i^{(N)}\right)_{N \in \N} \subset \widetilde{\K_i} $ converging to some $(s_1,\dots,s_i)$ for $N \rightarrow \infty$, we have that $f^{(N)}\left(s_1^{(N)},\dots,s_i^{(N)}\right) \rightarrow f(s_1,\dots,s_n)$ for $N \rightarrow \infty$.
Moreover, since by assumption $\mathfrak{F}_{i,j,k}|_{\widetilde{\K_i}}\subseteq C(\widetilde{\K_i})$ is compact, we can apply, e.g., \cite[Proposition 7.32, p. 148]{bertsekas2004stochastic} to $g|_{\widetilde{\K_i}}$ implying the continuity of 
\begin{equation*}
\widetilde{\K_i}\ni(s_1,\dots,s_i) \mapsto \inf_{f_{i,j,k} \in \mathfrak{F}_{i,j,k}|_{\widetilde{\K_i}}} H_{i,j,k}(s_1,\dots,s_i)\left(v_{j,k}(s_j)-f_{i,j,k}(s_1,\dots,s_i)\right) \text{ on } \widetilde{\K_i}\subset \R_+^i.
\end{equation*}
Since $\widetilde{\K_i}\subset \R_+^i$ was chosen arbitrarily 
we conclude also the continuity of
\[
\R_+^i\ni(s_1,\dots,s_i) \mapsto \inf_{f_{i,j,k} \in \mathfrak{F}_{i,j,k}} H_{i,j,k}(s_1,\dots,s_i)\left(v_{j,k}(s_j)-f_{i,j,k}(s_1,\dots,s_i)\right).
\]
Further, due to \eqref{eq_cond_clin_bounded}, we obtain 
\[
\sup_{(s_1,\dots,s_i)\in \R_+^i}\frac{\inf_{f_{i,j,k} \in \mathfrak{F}_{i,j,k}}\Psi^V_{(H_i),(H_{i,j,k}),(u_i)}(s,\left(f_{i,j,k}(s_1,\dots,s_i)\right)_{i,j,k}) }{1+\sum_{\ell=1}^i s_{\ell}}<\infty.
\]
Hence, using the definition of $\Psi^V_{(H_i),(H_{i,j,k}),(u_i)}$, we conclude \eqref{eq_phi_in_c_lin}.
Due to the validity of \eqref{eq_phi_in_c_lin}, when computing \eqref{eq_biconj_proof_inf} we get for every $\Q \in \mathcal{P}(\R_+^n)$ with finite first moments and $\Q(\K)=1$ that
\begin{align*}
&\operatorname{D}_{(\mathfrak{F}_{i,j,k})}^{*}(\Q)=\sup_{u_i \in C_{\operatorname{lin}}(\R_+,\R_+)} \sum_{i=1}^n\left(\int_{\K} u_i(s_i)\D\Q(s)-\int_{\R_+} u_i(s_i)\D\mu_i(s_i) \right) \\
+ &\sup_{H_i \in C_{b}(\R_+^i)} \sum_{i=1}^{n-1}\left(\int_{\K} H_i(s_1,\dots,s_i)(s_{i+1}-s_i)\D\Q(s)\right)\\
+&\sup_{H_{i,j,k} \in C_{b}(\R_+^i)} \sum_{i=1}^{n-1}\sum_{j=i+1}^{n}\sum_{k=1}^{N}\left(\int_{\K} \inf_{f_{i,j,k} \in \mathfrak{F}_{i,j,k}} H_{i,j,k}(s_1,\dots,s_i)(v_{j,k}(s_j)-f_{i,j,k}(s_1,\dots,s_i))\D\Q(s)\right).
\end{align*}
As in the proof of Theorem~\ref{thm_1}~(b) the first two summands vanish if and only if $\Q$ fulfils the associated marginal and martingale constraints. Moreover, the last summand is greater or equal to $0$ which can be seen through setting $H_{i,j,k} \equiv 0$ and we have that
\begin{align*}
0\leq&\sup_{H_{i,j,k} \in C_{b}(\R_+^i)} \sum_{i=1}^{n-1}\sum_{j=i+1}^{n}\sum_{k=1}^{N}\left(\int_{\K} \inf_{f_{i,j,k} \in \mathfrak{F}_{i,j,k}} H_{i,j,k}(s_1,\dots,s_i)(v_{j,k}(s_j)-f_{i,j,k}(s_1,\dots,s_i))\D\Q(s)\right)\\
\leq 
&\sup_{H_{i,j,k} \in C_{b}(\R_+^i)}\inf_{f_{i,j,k} \in \mathfrak{F}_{i,j,k}} \sum_{i=1}^{n-1}\sum_{j=i+1}^{n}\sum_{k=1}^{N}\left(\int_{\K}  H_{i,j,k}(s_1,\dots,s_i)(v_{j,k}(s_j)-f_{i,j,k}(s_1,\dots,s_i))\D\Q(s)\right)
\end{align*}
which vanishes if for all $i,j=1,\dots,n,k=1,\dots,N$ there exists some $f_{i,j,k} \in \mathfrak{F}_{i,j,k}$ such that $\E_\Q[v_{j,k}(S_{t_j})~|~S_{t_i},\dots,S_{t_1}]=f_{i,j,k}$ $\Q$-a.s.,\ and can be scaled infinitely large otherwise. This shows that the conjugate $\operatorname{D}_{(\mathfrak{F}_{i,j,k})}^{*}$ satisfies $\operatorname{D}_{(\mathfrak{F}_{i,j,k})}^{*}(\Q)=0$ if $\Q \in {\mathcal{Q}}_{(\mathfrak{F}_{i,j,k})}$.
To see that $\operatorname{D}_{(\mathfrak{F}_{i,j,k})}^{*}(\Q)=\infty$ if $\Q \not\in {\mathcal{Q}}_{(\mathfrak{F}_{i,j,k})}$ pick some $\Q \not\in {\mathcal{Q}}_{(\mathfrak{F}_{i,j,k})}$. From the arguments above we can assume w.l.o.g.\,that $\Q(\K)=1$ and that $\Q$ fulfills the marginal and martingale constraints. Then, by definition of $\Xi$ there exist some $i,j,k$ such that $\nexists f_{i,j,k} \in  \mathfrak{F}_{i,j,k}$ which satsifies
\begin{equation}\label{eq_proof_board_eq_00}
\E_\Q[v_{j,k}(S_{t_j})~|~ \mathcal{F}_{t_i}]=f_{i,j,k}~\Q \text{ - a.s.}
\end{equation}
Now note that for all $i,j,k$ and all $f_{i,j,k} \in  \mathfrak{F}_{i,j,k}$ we have
\begin{equation}\label{eq_proof_board_eq_0}
\sup_{H_{i,j,k} \in C_{b}(\R_+^i)} \int_{\K} H_{i,j,k}(s_1,\dots,s_i)(v_{j,k}(s_j)-f_{i,j,k}(s_1,\dots,s_i))\D\Q(s)= \begin{cases} 0 &\text{ if } \eqref{eq_proof_board_eq_00}\\
\infty &\text{ else}.
\end{cases}
\end{equation}
Therefore, since $\Q\not\in {\mathcal{Q}}_{(\mathfrak{F}_{i,j,k})}$ but satisfies the marginal and martingale constraints, there exists $i,j,k$ such that 
\begin{equation}\label{eq_proof_board_eq_1}
\inf_{f_{i,j,k} \in \mathfrak{F}_{i,j,k}}\sup_{H_{i,j,k} \in C_{b}(\R_+^i)} \int_{\K} H_{i,j,k}(s_1,\dots,s_i)(v_{j,k}(s_j)-f_{i,j,k}(s_1,\dots,s_i))\D\Q(s)=\infty.
\end{equation}
Now note that by the choice of $\Q$ we have
\begin{align}
\operatorname{D}_{(\mathfrak{F}_{i,j,k})}^{*}(\Q)&=\sup_{H_{i,j,k} \in C_{b}(\R_+^i)} \sum_{i=1}^{n-1}\sum_{j=i+1}^{n}\sum_{k=1}^{N}\int_{\K} \inf_{f_{i,j,k} \in \mathfrak{F}_{i,j,k}} H_{i,j,k}(s_1,\dots,s_i)(v_{j,k}(s_j)-f_{i,j,k}(s_1,\dots,s_i))\D\Q(s)\notag \\
&\hspace{0cm}=\sum_{i=1}^{n-1}\sum_{j=i+1}^{n}\sum_{k=1}^{N}\sup_{H_{i,j,k} \in C_{b}(\R_+^i)} \int_{\K} \inf_{f_{i,j,k} \in \mathfrak{F}_{i,j,k}} H_{i,j,k}(s_1,\dots,s_i)(v_{j,k}(s_j)-f_{i,j,k}(s_1,\dots,s_i))\D\Q(s)\notag \\
&\hspace{0cm}= \sum_{i=1}^{n-1}\sum_{j=i+1}^{n}\sum_{k=1}^{N}\sup_{H_{i,j,k} \in C_{b}(\K_i)}\int_{\K} \inf_{f_{i,j,k}|_{\K^i} \in \mathfrak{F}_{i,j,k}} H_{i,j,k}(s_1,\dots,s_i)(v_{j,k}(s_j)-f_{i,j,k}(s_1,\dots,s_i))\D\Q(s)
\label{eq_proof_board_eq_2},
\end{align}
where $\K_i \subset \R^i_+$ denotes the projection of $\K \subset \R^n_+$ onto the first $i$ components. Next, we have that as $\mathfrak{F}_{i,j,k}|_{\K_i} \subset C(\K_i)$ is compact and since for every $H_{i,j,k} \subset C_b(\K_i)$ the map
\[
C(\K_i)\supset \mathfrak{F}_{i,j,k}|_{\K_i} \ni f_{i,j,k} \mapsto \int_{\K}  H_{i,j,k}(s_1,\dots,s_i)(v_{j,k}(s_j)-f_{i,j,k}(s_1,\dots,s_i))\D\Q(s)\
\]
is continuous, we obtain that 
\begin{equation}\label{eq_proof_board_eq_3}
\operatorname{D}_{(\mathfrak{F}_{i,j,k})}^{*}(\Q)= \sum_{i=1}^{n-1}\sum_{j=i+1}^{n}\sum_{k=1}^{N}\sup_{H_{i,j,k}\in C_{b}(\K_i)} \inf_{f_{i,j,k}|_{\K^i} \in \mathfrak{F}_{i,j,k}} \int_{\K}  H_{i,j,k}(s_1,\dots,s_i)(v_{j,k}(s_j)-f_{i,j,k}(s_1,\dots,s_i))\D\Q(s).
\end{equation}
For each $i,j,k$ consider the map
\begin{equation}\label{eq_minimiax}
C(\K_i) \times \mathfrak{F}_{i,j,k}|_{\K_i} \ni (H_{i,j,k},f_{i,j,k}) \mapsto \int_{\K} H_{i,j,k}(s_1,\dots,s_i)(v_{j,k}(s_j)-f_{i,j,k}(s_1,\dots,s_i))\D\Q(s).
\end{equation}
By a minimax theorem, see, e.g. \cite[Theorem 4.2']{sion1958general}, applied to the map defined in \eqref{eq_minimiax}, we obtain from \eqref{eq_proof_board_eq_3} that
\begin{equation}\label{eq_proof_board_eq_4}
\operatorname{D}_{(\mathfrak{F}_{i,j,k})}^{*}(\Q)= \sum_{i=1}^{n-1}\sum_{j=i+1}^{n}\sum_{k=1}^{N}\inf_{f_{i,j,k}|_{\K^i} \in \mathfrak{F}_{i,j,k}} \sup_{H_{i,j,k}  \in C_{b}(\K_i)}\int_{\K}  H_{i,j,k}(s_1,\dots,s_i)(v_{j,k}(s_j)-f_{i,j,k}(s_1,\dots,s_i))\D\Q(s).
\end{equation}
Combining \eqref{eq_proof_board_eq_0} and \eqref{eq_proof_board_eq_1} we conclude that $\operatorname{D}_{(\mathfrak{F}_{i,j,k})}^{*}(\Q)=\infty$ as desired. This proves that the conjugate $\operatorname{D}_{(\mathfrak{F}_{i,j,k})}^{*}$ satisfies \eqref{eq_cases}.
\end{proof}

\begin{proof}[Proof of Remark~\ref{rem_eq_mot_with_and_without_constraints}]
Consider some super-replication strategy $\Psi^V_{(H_i),(H_{i,j,k}),(u_i)}$ such that 
\[
\Psi^V_{(H_i),(H_{i,j,k}),(u_i)}(s,p) \geq \Phi(s) \text{ for all }(s,p) \in \Xi.
\]
Then by \eqref{eq_sp_in_xi} assumed on $\Xi$ we have
directly
$\Psi^V_{(H_i),(H_{i,j,k}),(u_i)}(s,\widetilde{p}) \geq \Phi(s)$ for all $s \in \R_+^n $
and for the particular choice of $\widetilde{p}\in \R^{nN}$ with $\widetilde{p}_{i,j,k}=v_{j,k}(s_j)$. Moreover, since 
\[
\sum_{i=1}^{n-1}\sum_{j=i+1}^{n}\sum_{k=1}^{N} H_{i,j,k}(s_{1},\dots,s_{i})\left(v_{j,k}(s_{j})-\widetilde{p}_{i,j,k}\right) = 0,
\]
this implies that
 \begin{equation*}
\begin{aligned}
\Psi_{(H_i),(u_i)}(s):=\sum_{i=1}^n u_i(s_i)&+\sum_{i=1}^{n-1}H_i(s_1,\dots,s_i) (s_{{i+1}}-s_{i})\geq \Phi(s) \text{ for all } s \in \R_+^n.
\end{aligned}
\end{equation*}
Hence, we obtain
\begin{align*}
\operatorname{D}_{\Xi}(\Phi)&= \inf_{\substack{u_i\in C_{\operatorname{lin}}(\R_+,\R_+)\\ H_i,H_{i,j,k} \in C_b(\R^i):\\
 \Psi^V_{(H_i),(H_{i,j,k}),(u_i)}\geq \Phi}}\sum_{i=1}^n \int u_i(s_i)\D \mu_i(s_i)\\
 &\geq \inf_{\substack{u_i\in C_{\operatorname{lin}}(\R_+,\R_+)\\ H_i\in C_b(\R^i):\\
 \Psi_{(H_i),(u_i)} \geq \Phi}}\sum_{i=1}^n \int u_i(s_i)\D \mu_i(s_i) \\
 &= \sup_{\Q \in \mathcal{M}(\mu_1,\dots,\mu_n)}\int_{\R_+^n}\Phi(s)\D\Q(s)
\end{align*}
where the last equality is the martingale optimal transport duality from \cite[Corollary 1.2.]{beiglbock2013model}. The reverse inequality follows immediately by definition.
\end{proof}

\begin{proof}[Proof of Theorem~\ref{thm_critical_points}]
We first prove the assertion from (a). W.l.o.g. assume \eqref{eq_thm_condition1} holds true, as the case \eqref{eq_thm_condition2} can be argued analogously. First we claim that
\[
\widehat{\mathcal{M}}_V(\Omega,\mu_1,\dots,\mu_n) \cap \widehat{\mathcal{M}}_V(\Xi_{(\underline{p}_{i,j,k},\overline{p}_{i,j,k})},\mu_1,\dots,\mu_n)=\emptyset.
\]
Assume by contradiction that there exists $\Q \in \widehat{\mathcal{M}}_V(\Omega,\mu_1,\dots,\mu_n)\cap \widehat{\mathcal{M}}_V(\Xi_{(\underline{p}_{i,j,k},\overline{p}_{i,j,k})},\mu_1,\dots,\mu_n) $. In particular, we have for this $\Q$ that $\Q\left(\operatorname{P}_{t_i}(v_{j,k})\in [\underline{p}_{i,j,k},\overline{p}_{i,j,k}]\right)=1$ for all $i,j \in \{1,\dots,n\},k\in \{1,\dots,N\}$. Set $\widetilde{A}= A \cap \{(s,p) \in \Omega~|~p_{i,j,k} \in [\underline{p}_{i,j,k},\overline{p}_{i,j,k}]\}$ where $i,j,k$ are the indices and $A$ is the set corresponding to \eqref{eq_thm_condition1}. Then, by validity of \eqref{eq_thm_condition1}, we obtain the following inequality
$$
\int_{\widetilde{A}}(v_{j,k}(s_j)-p_{i,j,k})\D\Q(s,p)\leq \int_{\widetilde{A}}(v_{j,k}(s_j)-\underline{p}_{i,j,k}(s_1,\dots,s_i))\D\Q(s,p)<0
$$
which contradicts the definition of $\operatorname{P}_{t_i}(v_{j,k})$ which coincides $\Q$-a.s. with the $\mathcal{F}_{t_i}$-conditional expectation of $v_{j,k}(S_{t_j})$.
Thus 
$$
\widehat{\mathcal{M}}_V(\Omega,\mu_1,\dots,\mu_n) \cap \widehat{\mathcal{M}}_V(\Xi_{(\underline{p}_{i,j,k},\overline{p}_{i,j,k})},\mu_1,\dots,\mu_n)=\emptyset.
$$
Moreover, by Theorem~\ref{thm_1}~(b), there exists some $\overline{\Q}\in \widehat{\mathcal{M}}_V(\Xi_{(\underline{p}_{i,j,k},\overline{p}_{i,j,k})},\mu_1,\dots,\mu_n)$. Therefore, as ${\mathcal{M}}_V(\Xi_{(\underline{p}_{i,j,k},\overline{p}_{i,j,k})},\mu_1,\dots,\mu_n) \subset {\mathcal{M}}_V(\Omega,\mu_1,\dots,\mu_n)$, we have for all $\Q \in \widehat{\mathcal{M}}_V(\Omega,\mu_1,\dots,\mu_n)$ that
\begin{align*}
&\operatorname{P}_{\Xi_{(\underline{p}_{i,j,k},\overline{p}_{i,j,k})}}(\Phi)= \int \Phi \D \overline{\Q}<\int \Phi \D \Q = \operatorname{P}_\Omega(\Phi).
\end{align*}
On the other hand, if neither \eqref{eq_thm_condition1} nor \eqref{eq_thm_condition2} hold true, then there exists some measure $\Q \in \widehat{\mathcal{M}}_V(\Omega,\mu_1,\dots,\mu_n)$ such that 
\[
\underline{p}_{i,j,k} \leq \E_\Q\left[v_{j,k}(S_{t_{j}})~\middle|~\mathcal{F}_{t_{i}}\right] \leq \overline{p}_{i,j,k}~\Q\text{-a.s.} \text{ for all } i,j,k.
\]
Hence $\Q \in {\mathcal{M}}_V(\Xi_{(\underline{p}_{i,j,k},\overline{p}_{i,j,k})},\mu_1,\dots,\mu_n)$ and consequently
\[
\operatorname{P}_{\Xi_{(\underline{p}_{i,j,k},\overline{p}_{i,j,k})}}(\Phi)\geq  \int \Phi \D {\Q}= \operatorname{P}_\Omega(\Phi),
\]
which in turn implies equality.
\\
For the assertion from (b) one can show analogously that if \eqref{eq_thm_condition3} holds, then 
\[
\widehat{\mathcal{M}}_V(\Xi_{(\underline{p}_{i,j,k},\overline{p}_{i,j,k})},\mu_1,\dots,\mu_n)  \cap\widehat{\mathcal{M}}_V(\Xi_{(\underline{p}_{i,j,k}+\varepsilon,\overline{p}_{i,j,k})},\mu_1,\dots,\mu_n) = \emptyset
\]
and conclude that for all $\Q \in \widehat{\mathcal{M}}_V(\Xi_{(\underline{p}_{i,j,k},\overline{p}_{i,j,k})},\mu_1,\dots,\mu_n)$
\[
\operatorname{P}_{\Xi_{(\underline{p}_{i,j,k}+\varepsilon,\overline{p}_{i,j,k})}}(\Phi)< \int \Phi \D \Q(s,p) = \operatorname{P}_{\Xi_{(\underline{p}_{i,j,k},\overline{p}_{i,j,k})}}(\Phi).
\]
For the reverse direction we remark that if \eqref{eq_thm_condition3} does not hold, then there exists some $\Q \in \widehat{\mathcal{M}}_V(\Xi_{(\underline{p}_{i,j,k},\overline{p}_{i,j,k})},\mu_1,\dots,\mu_n)$ such that 
\begin{equation}\label{eq_proof_<1}
\underline{p}_{i,j,k} +\varepsilon \leq \E_\Q\left[v_{j,k}(S_{t_{j}})~\middle|~\mathcal{F}_{t_{i}}\right]\leq \overline{p}_{i,j,k}~\Q\text{-a.s.} \text{ for all } i,j,k.
\end{equation}
Hence 
\[
\operatorname{P}_{\Xi_{(\underline{p}_{i,j,k}+\varepsilon,\overline{p}_{i,j,k})}}(\Phi)\geq  \int \Phi \D {\Q}= \operatorname{P}_\Omega(\Phi),
\]
which in turn implies equality.
The assertion from (c) follows in the same way as in the proof of (b).
\end{proof}

\begin{proof}[Proof of Theorem~\ref{thm_infinity_N}]
The proof is analogue to the proof of Remark~\ref{rem_thm_1}~(c). Analogue to equation \eqref{eq_dual_infinity} we see that 
\begin{equation}\label{eq_proof_fijk_xigeq0}
\widetilde{\operatorname{D}}_{(\widetilde{\mathfrak{F}}_{i,j,k})}^*(\Q) = \infty \text{ if } \Q(\widetilde{\Xi}^c)>0.
\end{equation} 
Moreover, when computing the convex conjugate of $\widetilde{\operatorname{D}}_{(\widetilde{\mathfrak{F}}_{i,j,k})}(\Phi)$ we obtain for every $\Q \in \mathcal{P}(\R_+^n)$ with finite first moments satisfying $\Q(\widetilde{\Xi})=1$ that
\begin{align*}
&\widetilde{\operatorname{D}}_{(\widetilde{\mathfrak{F}}_{i,j,k})}^*(\Q)=\sup_{u_i \in C_{\operatorname{lin}}(\R_+,\R_+)} \sum_{i=1}^n\left(\int_{\widetilde{\Xi}} u_i(s_i)\D\Q(s)-\int_{\R_+} u_i(s_i)\D\mu_i(s_i) \right) \\
+ &\sup_{H_i \in C_{b}(\R_+^i)} \sum_{i=1}^{n-1}\left(\int_{\widetilde{\Xi}} H_i(s_1,\dots,s_i)(s_{i+1}-s_i)\D\Q(s)\right)\\
+&\sup_{H_{i,j,k} \in C_{b}(\R_+^i):\atop \eqref{eq_finitely_many_nonzero} \text{ holds}} \sum_{i=1}^{n-1}\sum_{j=i+1}^{n} \sum_{k \in \mathcal{I}_{\widetilde{V}}}\left(\int_{\widetilde{\Xi}} \inf_{f_{i,j,k} \in \widetilde{\mathfrak{F}}_{i,j,k}} H_{i,j,k}(s_1,\dots,s_i)(v_{j,k}(s_j)-f_{i,j,k}(s_1,\dots,s_i))\D\Q(s)\right).
\end{align*}

The first two suprema vanish if and only if $\Q$ fulfils the corresponding martingale and marginal properties, whereas, by the same minimax-argument as in the proof of Remark~\ref{rem_thm_1}~(c), the last supremum vanishes if and only if for all $i,j=1,\dots,n$ and all $v_{j,k} \in \widetilde{V} \subseteq C_{\operatorname{lin}}(\R_+,\R_+)$ there exists $f_{i,j,k} \in \widetilde{\mathfrak{F}}_{i,j,k}$ such that 
\[
\E_\Q[v_{j,k}(S_{t_j})~|~S_{t_i},\dots,S_{t_1}]=f_{i,j,k} ~\Q\text{-a.s.}
\]
This means, together with \eqref{eq_proof_fijk_xigeq0} that $\widetilde{\operatorname{D}}_{(\widetilde{\mathfrak{F}}_{i,j,k})}^*(\Q)=0$ if and only if $\Q \in \widetilde{\mathcal{Q}}_{(\widetilde{\mathfrak{F}}_{i,j,k})}$, and otherwise $\widetilde{\operatorname{D}}_{(\widetilde{\mathfrak{F}}_{i,j,k})}^*(\Q)$ becomes infinitely large. Hence we conclude the results by the biconjuagte representation for $\widetilde{\operatorname{D}}_{(\widetilde{\mathfrak{F}}_{i,j,k})}$ similar to \eqref{eq_convex_duality}.
\end{proof}

\begin{proof}[Proof of Remark~\ref{rem_examples_fijk}]
Condition \eqref{eq_linear_growth_N_inf} is fulfilled since for all $f \in\widetilde{\mathfrak{F}}_{i,j,k}$ and all $(s_1,\dots,s_i) \in \R_+^i$ the Lipschitz property ensures that
\begin{equation}\label{eq_ineq_1}
\left|f(s_1,\dots,s_i)\right|\leq |g(s_i)-g(0)|+|g(0)|\leq s_i.
\end{equation}
To see that for every compact set $\K \subset \R_+^i$ the set $\widetilde{\mathfrak{F}}_{i,j,k}$ is compact when restricted onto $\K$, pick a sequence $(f_{i,j,k}^{(N)})_{N \in \N}$ with $f_{i,j,k}^{(N)} \in \widetilde{\mathfrak{F}}_{i,j,k}$ for all $N \in \N$. Then we obtain for all $N \in \N$ a representation $f_{i,j,k}^{(N)}(s_1,\dots,s_i)=g^{(N)}(s_i)$ for some $1$-Lipschitz function $g^{(N)}$. By the $1$-Lipschitz property of $g^{(N)}$ the sequence $(g^{(N)})_{N \in \N}$ is uniformly equicontinuous and pointwise bounded according to \eqref{eq_ineq_1}.

Thus, the Arzelà–Ascoli theorem implies the existence of a uniformly convergent subsequence (labelled identically) with $g^{(N)}\rightarrow g$ for $N \rightarrow \infty$ for some function $g$. Then $g$ is $1$-Lipschitz as we have for all $x,y \in \R_+$ that $|g(x)-g(y)|=\lim_{N \rightarrow \infty}|g^{(N)}(x)-g^{(N)}(y)|\leq |x-y|$. Further, we have $g(0)=\lim_{N \rightarrow \infty} g^{(N)}(0)=0$. It remains to show that $g$ admits a representation of the form
\[
\E_{\Q}[v_{j,k}(S_{t_j})~|~S_{t_i}]=g~\Q\text{-a.s.}\text{ for some } \Q \in \mathcal{M}(\mu_1,\dots,\mu_n).
\]
By definition of $\widetilde{\mathfrak{F}}_{i,j,k}$, we have for all $N \in \N$ the representation
\begin{equation}\label{eq_rem_representation}
g^{(N)}=\E_{\Q^{(N)}}[v_{j,k}(S_{t_j})~|~S_{t_i}] ~\Q^{(N)}\text{-a.s. for some } \Q^{(N)}\in \mathcal{M}(\mu_1,\dots,\mu_n).
\end{equation}
Then by the weak compactness of $\mathcal{M}(\mu_1,\dots,\mu_n)$ there exists some subsequence of $\left(\Q^{(N)}\right)_{N \in \N} \subseteq \mathcal{M}(\mu_1,\dots,\mu_n)$ (denoted identically) converging weakly to some $\Q \in \mathcal{M}(\mu_1,\dots,\mu_n)$. Similar to \cite[Lemma 3.3.]{sester2020robust}  we obtain for all $\Delta \in C_b(\R_+)$ and all $i,j=1,\dots,n$ that
\begin{align}
&\int_{\R_+^n} \Delta(s_i)\lim_{N\rightarrow \infty}g^{(N)}(s_i) \D \Q(s_1,\dots,s_i) \label{eq_rem_1}\\
=&\lim_{N\rightarrow \infty}\int_{\R_+^n} \Delta(s_i)g^{(N)}(s_i) \D \Q(s_1,\dots,s_i) \label{eq_rem_after_1_added}\\
=&\lim_{N\rightarrow \infty}\int_{\R_+^n} \Delta(s_i)g^{(N)}(s_i) \D \Q^{(N)}(s_1,\dots,s_i)\label{eq_rem_2} \\
=&\lim_{N\rightarrow \infty}\int_{\R_+^n} \Delta(s_i)v_{j,k}(s_j) \D \Q^{(N)}(s_1,\dots,s_i)\label{eq_rem_3}\\
=&\int_{\R_+^n} \Delta(s_i)v_{j,k}(s_j) \D \Q(s_1,\dots,s_i)\label{eq_rem_4},
\end{align}
where the equality between \eqref{eq_rem_1} and \eqref{eq_rem_after_1_added} follows due to dominated convergence (which can be seen through the validity of \eqref{eq_linear_growth_N_inf} and by $\Q \in \mathcal{M}(\mu_1,\dots,\mu_n)$), the equality between \eqref{eq_rem_after_1_added} and \eqref{eq_rem_2} holds since $\Q^{(N)}\circ S_{t_i}^{-1}=\Q\circ S_{t_i}^{-1}$, the equality between \eqref{eq_rem_2} and \eqref{eq_rem_3} is a consequence of \eqref{eq_rem_representation}, and \eqref{eq_rem_4} follows from \cite[Lemma 2.2.]{beiglbock2013model}.
Hence, we conclude that
\[
\lim_{N \rightarrow \infty}g^{(N)}=\E_{\Q}[v_{j,k}(S_{t_j})~|~S_{t_i}]~\Q\text{-a.s.}
\]
and thus $g = \E_{\Q}[v_{j,k}(S_{t_j})~|~S_{t_i}]~{\Q}$-a.s.
\end{proof}
\vspace{0.3cm}
\section*{Acknowledgements}
Financial support of the NAP Grant \emph{Machine Learning based Algorithms in Finance and Insurance} is gratefully acknowledged. We thank two anonymous referees for extraordinary carefully reading the manuscript and for useful comments that led to an improvement of the paper.
Further we acknowledge the Singapore National Supercomputing Centre (NSCC) which provided computing power to conduct the numerical examples for the research.
\appendix

\appendix

\section{Extensions}\label{sec_extension}
In this section we discuss various extensions of the presented results in Section~\ref{sec_setup_main_results}. We extend our considerations to multiple securities, market frictions such as transaction costs and the inclusion of different kinds of (path-dependent) options for dynamic trading.

\subsection{Transaction costs}\label{sec_frictions}

Since the considered strategies include an additional dynamic trading component which may cause a significant additional amount of transaction costs, it is important to discuss how to incorporate transactions costs of the form considered for example in \cite[Section 3.1.]{cheridito2017duality}. For simplicity, we stick to the setting described in \cite{cheridito2017duality} where one considers only a finite amount of staticly traded call options instead of general European options. When considering transaction costs, the profits of the considered strategies (without pricing rules) will be reduced and change for $(s,p) \in \Xi$ to:

\begin{equation}\label{strategy_transaction_costs}
\begin{aligned}
&\sum_{i=1}^n\sum_{j=1}^{M_i} \theta_{i,j} (s_i-K_{i,j})^+-h_{i,j}(\theta_{i,j}) \\
+&\sum_{i=1}^{n-1}H_i(s_1,\dots,s_i) (s_{i+1}-s_i)-g_{i}^{\operatorname{stock}}(\Delta H_{i+1} \cdot s_i)\\
+&\sum_{i=1}^{n-1}\sum_{j=i+1}^{n}\sum_{k=1}^{N} \bigg(H_{i,j,k}(s_1,\dots,s_i)\left(v_{j,k}(s_j)-p_{i,j,k}\right)-g_{i}^{\operatorname{option}}\left(\Delta H_{i+1,j,k} \cdot p_{i,j,k}\right)\bigg)
\end{aligned}
\end{equation}
with $\theta_{i,j}$, $K_{i,j}\in \R_+$, $M_i \in \N$, $\Delta H_{i+1} = H_{i+1}-H_i$, $\Delta H_{i+1,j,k} = H_{i+1,j,k}-H_{i,j,k}$, and $h_{i,j}$, $g_{i}^{\operatorname{stock}}$, $g_{i}^{\operatorname{option}}$ real-valued functions associated to the respective trading positions. \\
In the case of proportional transaction costs one has $h_{i,j}(\theta_{i,j})=\theta_{i,j}^+ h_{i,j}^+-\theta_{i,j}^- h_{i,j}^-$ where $h_{i,j}^+, h_{i,j}^-$ denote ask and bid prices of the considered call options. Moreover, we have $g_{i}^{\operatorname{stock}}(x)=\varepsilon_i^{\operatorname{stock}} |x|$ and  $g_{i}^{\operatorname{option}}(x)=\varepsilon_i^{\operatorname{option}} |x|$ for some $\varepsilon_i^{\operatorname{stock}},\varepsilon_i^{\operatorname{option}} \geq 0$.
In this case combining our duality argument with the argumentation of \cite{cheridito2017duality}, we obtain that minimizing prices of \eqref{strategy_transaction_costs}-type super-replication strategies of $\Phi(S)$ for $\Phi \in C_{\operatorname{lin}}\left(\R_+^n,\R_+\right)$ is equivalent to
\[
\sup_{\Q \in \mathcal{M}^{\operatorname{prop}}}\int_{\Omega} \Phi(s) \D \Q(s,p),
\]
where $\mathcal{M}^{\operatorname{prop}}$ is the set of all probability measures $\Q$ on $\Omega$ with
\begin{itemize}
\item[(i)]$(1-\varepsilon_i^{\operatorname{stock}})S_{t_i} \leq \E_\Q[S_{t_{i+1}}|\mathcal{F}_{t_i}] \leq (1+\varepsilon_i^{\operatorname{stock}})S_{t_i}~\Q\text{-a.s.} $ for all $i=1,\dots,n$,
\item[(ii)] $(1-\varepsilon_i^{\operatorname{option}})\operatorname{P}_{t_i}(v_{j,k}) \leq \E_\Q[v_{j,k}(S_{t_j})|\mathcal{F}_{t_i}] \leq (1+\varepsilon_i^{\operatorname{option}})\operatorname{P}_{t_i}(v_{j,k}) ~\Q\text{-a.s.}$ for all $i,j=1,\dots,n$, $k=1,\dots.N$,
\item[(iii)] $h_{i,j}^- \leq \E_\Q[(S_{t_i}-K_{i,j})^+] \leq h_{i,j}^+~\Q\text{-a.s.}$ for all $i=1,\dots,n, j=1,\dots,M_i$,
\item[(iv)] $\Q\left(\Xi\right)=1$.
\end{itemize}
This means that on the primal side we obtain an optimization problem over a set of measures with relaxed inequality constraints which will eventually lead to higher maximal prices compared with the formulation without transaction costs.

\subsection{Multiple securities}\label{sec_multiple_securities}
The considerations from Section~\ref{sec_setup_main_results} can be extended straightforward to a high-dimensional market in which we consider $d \geq 2$ stocks, tradable at $n\in \N$ future times. In this case one considers trading strategies of the form
\begin{equation*}
\begin{aligned}
\sum_{l=1}^d\bigg(\sum_{i=1}^n u_i^l(s_i^l)+\sum_{i=1}^{n-1}H_i^l(s_1,\dots,s_i) (s_{{i+1}}^l-s_{i}^l)+\sum_{i=1}^{n-1}\sum_{j=i+1}^{n}\sum_{k=1}^{N} H_{i,j,k}^l(s_{1},\dots,s_{i})\left(v_{j,k}^l(s_{j})-p^l_{i,j,k}\right)\bigg).
\end{aligned}
\end{equation*}
for $(s_1,\dots,s_n)=(s_1^1,\dots,s_n^d)\in \R^{nd}$ and $p=(p_{1,1,1}^1,\dots,p_{n,n,N}^d) \in \R^{nNd}$. We stress that the strategies $H_i^l$ and $H_{i,j,k}^l$ are for all $l =1,\dots,d$ allowed to depend on the price paths of all of the other securities under considerations, i.e., all available information is taken into account for trading. On the primal side this corresponds to joint martingale properties of the form
\begin{align*}
&\E_\Q\left[S_{t_{i+1}}^l~\middle|~S_{t_i}^1,\dots,S_{t_i}^d,\dots,S_{t_1}^d\right]=S_{t_{i}}^l~\Q\text{-a.s.},\\
&\E_\Q\left[v_{j,k}^l(S_{t_j}^l)~\middle|~S_{t_i}^1,\dots,S_{t_i}^d,\dots,S_{t_1}^d\right]=\operatorname{P}_{t_{i}}(v_{j,k}^l)~\Q\text{-a.s.}
\end{align*}
for all $i=1,\dots,n,~j=i+1,\dots,n,~l=1,\dots,d.$

\subsection{Path-dependent traded options}\label{sec_path_dependent}
From a mathematical point of view there is no need to restrict the considerations to the case of dynamic trading in European options, i.e., to options where the associated payoff function only depends on a sole value of an underlying security. However, the assumption to allow dynamic trading over time requires from a practical point of view that the involved option is traded in a sufficiently liquid amount over time. This is very often only fulfilled for specific European options such as call and put options.
However, if the liquidity of options is ensured, it is also thinkable to allow for trading in other kind of options that are possibly depending on the whole path of an underlying security. If $v_{j,k}$ depends on the whole path until time $t_j$, then we substitute \eqref{eq_pti_martingale} 
by 
 \begin{equation}
\begin{aligned}
\operatorname{P}_{t_i}\left(v_{j,k}\right)=\E_\Q\left[v_{j,k}(S_{t_1},\dots,S_{t_j})\middle|\mathcal{F}_{t_i}\right]~\Q\text{-a.s.}
\end{aligned}
\end{equation}
and accordingly on the dual side the term expressing the dynamic position in the traded options changes to
\begin{equation}
\sum_{i=1}^{n-1}\sum_{j=i+1}^{n}\sum_{k=1}^{N} H_{i,j,k}(s_{1},\dots,s_{i})\left(v_{j,k}(s_1,\dots,s_{j})-p_{i,j,k}\right).
\end{equation}
Similarly one can include dynamically traded basket options, i.e., (possibly path-dependent) options that depend on a multitude of underlying securities.
\bibliographystyle{plain} 
\bibliography{literature}

\begin{thebibliography}{10}

\bibitem{acciaio2016model}
Beatrice Acciaio, Mathias Beiglb{\"o}ck, Friedrich Penkner, and Walter
  Schachermayer.
\newblock A model-free version of the fundamental theorem of asset pricing and
  the super-replication theorem.
\newblock {\em Mathematical Finance}, 26(2):233--251, 2016.

\bibitem{aksamit2019robust}
Anna Aksamit, Shuoqing Deng, Jan Ob{\l}{\'o}j, and Xiaolu Tan.
\newblock The robust pricing--hedging duality for {A}merican options in
  discrete time financial markets.
\newblock {\em Mathematical Finance}, 29(3):861--897, 2019.

\bibitem{alfonsi2019sampling}
Aur{\'e}lien Alfonsi, Jacopo Corbetta, and Benjamin Jourdain.
\newblock Sampling of one-dimensional probability measures in the convex order
  and computation of robust option price bounds.
\newblock {\em International Journal of Theoretical and Applied Finance},
  22(03):1950002, 2019.

\bibitem{ansari2020improved}
Jonathan Ansari, Eva L\"utkebohmert, Ariel Neufeld, and Julian Sester.
\newblock Improved robust price bounds for multi-asset derivatives under
  market-implied dependence information.
\newblock {\em Preprint}, 2021.

\bibitem{baker2012martingales}
David Baker.
\newblock Martingales with specified marginals.
\newblock {\em Theses, Universit{\'e} Pierre et Marie Curie-Paris VI}, 2012.

\bibitem{bartl2019robust}
Daniel Bartl, Patrick Cheridito, and Michael Kupper.
\newblock Robust expected utility maximization with medial limits.
\newblock {\em Journal of Mathematical Analysis and Applications},
  471(1-2):752--775, 2019.

\bibitem{bartl2020pathwise}
Daniel Bartl, Michael Kupper, and Ariel Neufeld.
\newblock Pathwise superhedging on prediction sets.
\newblock {\em Finance and Stochastics}, 24(1):215--248, 2020.

\bibitem{bartl2019duality}
Daniel Bartl, Michael Kupper, David~J Pr{\"o}mel, and Ludovic Tangpi.
\newblock Duality for pathwise superhedging in continuous time.
\newblock {\em Finance and Stochastics}, 23(3):697--728, 2019.

\bibitem{bayraktar2017arbitrage}
Erhan Bayraktar and Zhou Zhou.
\newblock On arbitrage and duality under model uncertainty and portfolio
  constraints.
\newblock {\em Mathematical Finance}, 27(4):988--1012, 2017.

\bibitem{beiglbock2013model}
Mathias Beiglb{\"o}ck, Pierre Henry-Labord{\`e}re, and Friedrich Penkner.
\newblock Model-independent bounds for option prices—a mass transport
  approach.
\newblock {\em Finance and Stochastics}, 17(3):477--501, 2013.

\bibitem{beiglbock2017complete}
Mathias Beiglb{\"o}ck, Marcel Nutz, and Nizar Touzi.
\newblock Complete duality for martingale optimal transport on the line.
\newblock {\em The Annals of Probability}, 45(5):3038--3074, 2017.

\bibitem{bertsekas2004stochastic}
Dimitri~P Bertsekas and Steven Shreve.
\newblock {\em Stochastic optimal control: the discrete-time case}.
\newblock Academic Press, Inc., 1978.

\bibitem{biagini2004super}
Sara Biagini and Marco Frittelli.
\newblock On the super replication price of unbounded claims.
\newblock {\em The Annals of Applied Probability}, 14(4):1970--1991, 2004.

\bibitem{bouchard2015arbitrage}
Bruno Bouchard and Marcel Nutz.
\newblock Arbitrage and duality in nondominated discrete-time models.
\newblock {\em The Annals of Applied Probability}, 25(2):823--859, 2015.

\bibitem{breeden1978prices}
Douglas~T Breeden and Robert~H Litzenberger.
\newblock Prices of state-contingent claims implicit in option prices.
\newblock {\em Journal of business}, pages 621--651, 1978.

\bibitem{BurzoniFritelliMaggisAAP}
Matteo Burzoni, Marco Frittelli, and Marco Maggis.
\newblock Model-free superhedging duality.
\newblock {\em Ann. Appl. Probab.}, 27(3):1452--1477, 2017.

\bibitem{carr2001towards}
Peter Carr and Dilip Madan.
\newblock Towards a theory of volatility trading.
\newblock {\em Option Pricing, Interest Rates and Risk Management, Handbooks in
  Mathematical Finance}, pages 458--476, 2001.

\bibitem{cheridito2020martingale}
Patrick Cheridito, Matti Kiiski, David~J Pr{\"o}mel, and H~Mete Soner.
\newblock Martingale optimal transport duality.
\newblock {\em Mathematische Annalen}, pages 1--28, 2020.

\bibitem{cheridito2017duality}
Patrick Cheridito, Michael Kupper, and Ludovic Tangpi.
\newblock Duality formulas for robust pricing and hedging in discrete time.
\newblock {\em SIAM Journal on Financial Mathematics}, 8(1):738--765, 2017.

\bibitem{cvitanic1999closed}
Jak{\v{s}}a Cvitani{\'c}, Huyen Pham, and Nizar Touzi.
\newblock A closed-form solution to the problem of super-replication under
  transaction costs.
\newblock {\em Finance and {S}tochastics}, 3(1):35--54, 1999.

\bibitem{dantzig1998linear}
George~Bernard Dantzig.
\newblock {\em Linear programming and extensions}, volume~48.
\newblock Princeton university press, 1998.

\bibitem{aquino2019bounds}
Luca De~Gennara~Aquino and Carole Bernard.
\newblock Bounds on multi-asset derivatives via neural networks.
\newblock {\em International Journal of Theoretical and Applied Finance},
  23(08):2050050, 2020.

\bibitem{dolinsky2018super}
Yan Dolinsky and Ariel Neufeld.
\newblock Super-replication in fully incomplete markets.
\newblock {\em Mathematical Finance}, 28(2):483--515, 2018.

\bibitem{dolinsky2014martingale}
Yan Dolinsky and H~Mete Soner.
\newblock Martingale optimal transport and robust hedging in continuous time.
\newblock {\em Probability Theory and Related Fields}, 160(1-2):391--427, 2014.

\bibitem{eberlein1997range}
Ernst Eberlein and Jean Jacod.
\newblock On the range of options prices.
\newblock {\em Finance and Stochastics}, 1(2):131--140, 1997.

\bibitem{eckstein2019computation}
Stephan Eckstein and Michael Kupper.
\newblock Computation of optimal transport and related hedging problems via
  penalization and neural networks.
\newblock {\em Applied Mathematics \& Optimization}, 83(2):639--667, 2021.

\bibitem{eckstein2020martingale}
Stephan Eckstein and Michael Kupper.
\newblock Martingale transport with homogeneous stock movements.
\newblock {\em Quantitative Finance}, 21(2):271--280, 2021.

\bibitem{follmer2016stochastic}
Hans F{\"o}llmer and Alexander Schied.
\newblock {\em Stochastic finance: an introduction in discrete time}.
\newblock Walter de Gruyter, 2016.

\bibitem{frey1999bounds}
R{\"u}diger Frey and Carlos~A Sin.
\newblock Bounds on {E}uropean option prices under stochastic volatility.
\newblock {\em Mathematical Finance}, 9(2):97--116, 1999.

\bibitem{guo2019computational}
Gaoyue Guo and Jan Ob{\l}{\'o}j.
\newblock Computational methods for martingale optimal transport problems.
\newblock {\em The Annals of Applied Probability}, 29(6):3311--3347, 2019.

\bibitem{henry2013automated}
Pierre Henry-Labord{\`e}re.
\newblock Automated option pricing: Numerical methods.
\newblock {\em International Journal of Theoretical and Applied Finance},
  16(08):1350042, 2013.

\bibitem{henry2019martingale}
Pierre Henry-Labordere.
\newblock (martingale) optimal transport and anomaly detection with neural
  networks: A primal-dual algorithm.
\newblock {\em Available at SSRN 3370910}, 2019.

\bibitem{hou2018robust}
Zhaoxu Hou and Jan Ob{\l}{\'o}j.
\newblock Robust pricing--hedging dualities in continuous time.
\newblock {\em Finance and Stochastics}, 22(3):511--567, 2018.

\bibitem{kellerer1972markov}
Hans~G Kellerer.
\newblock Markov-{K}omposition und eine {A}nwendung auf {M}artingale.
\newblock {\em Mathematische Annalen}, 198(3):99--122, 1972.

\bibitem{levental1997possibility}
Shlomo Levental and Anatolii~V Skorohod.
\newblock On the possibility of hedging options in the presence of transaction
  costs.
\newblock {\em The Annals of Applied Probability}, 7(2):410--443, 1997.

\bibitem{lutkebohmert2019tightening}
Eva L{\"u}tkebohmert and Julian Sester.
\newblock Tightening robust price bounds for exotic derivatives.
\newblock {\em Quantitative Finance}, 19(11):1797--1815, 2019.

\bibitem{mykland2003financial}
Per~Aslak Mykland.
\newblock Financial options and statistical prediction intervals.
\newblock {\em The Annals of Statistics}, 31(5):1413--1438, 2003.

\bibitem{neufeld2018buy}
Ariel Neufeld.
\newblock Buy-and-hold property for fully incomplete markets when
  super-replicating markovian claims.
\newblock {\em International Journal of Theoretical and Applied Finance},
  21(08):1850051, 2018.

\bibitem{neufeld2020model}
Ariel Neufeld, Antonis Papapantoleon, and Qikun Xiang.
\newblock Model-free bounds for multi-asset options using option-implied
  information and their exact computation.
\newblock {\em arXiv preprint arXiv:2006.14288}, 2020.

\bibitem{sester2020robust}
Julian Sester.
\newblock Robust price bounds for derivative prices in markovian models.
\newblock {\em International Journal of Theoretical and Applied Finance},
  23(3):2050015, 2020.

\bibitem{sion1958general}
Maurice Sion.
\newblock On general minimax theorems.
\newblock {\em Pacific Journal of mathematics}, 8(1):171--176, 1958.

\bibitem{strassen1965existence}
Volker Strassen.
\newblock The existence of probability measures with given marginals.
\newblock {\em The Annals of Mathematical Statistics}, 36(2):423--439, 1965.

\end{thebibliography}
\end{document}